\documentclass{LMCS}

\def\dOi{10(2:1)2014}
\lmcsheading%
{\dOi}
{1--47}
{}
{}
{Jan.~13, 2013}
{Apr.~\phantom04, 2014}
{}

\ACMCCS{[{\bf Theory of computation}]: Formal languages and automata
  theory---Formalisms---Rewrite systems}


\usepackage[utf8]{inputenc}
\usepackage[english]{babel}
\usepackage{graphicx}
\usepackage{hyperref}
\usepackage[matrix,arrow,curve]{xy}
\CompileMatrices
\usepackage{wrapfig}
\usepackage{shortcuts}
\usepackage{macros}

\newcommand{\strid}[1]{\scpdfinput{0.9}{#1.pdf}}

\renewcommand{\leq}{\leqslant}

\newenvironment{definition}{\begin{defi}}{\end{defi}}
\newenvironment{lemma}{\begin{lem}}{\end{lem}}
\newenvironment{example}{\begin{exa}}{\end{exa}}
\newenvironment{remark}{\begin{rem}}{\end{rem}}
\newenvironment{proposition}{\begin{prop}}{\end{prop}}

\title[Towards 3-Dimensional Rewriting Theory]
      {Towards 3-Dimensional Rewriting Theory}

\thanks{This work was partially funded by the french ANR project CATHRE ANR-13-BS02-0005-02.}

\author[S.~Mimram]{Samuel Mimram}
\address{CEA, LIST, LMeASI, 91191 Gif-sur-Yvette, France.}
\email{\href{mailto:samuel.mimram@cea.fr}{samuel.mimram@cea.fr}}
\hypersetup{
  pdftitle={\csname @title\endcsname},
  pdfauthor={Samuel Mimram},
  unicode=true,
}
\renewcommand{\C}{\mathcal{C}}

\begin{document}
\begin{abstract}
  String rewriting systems have proved very useful to study monoids. In good
  cases, they give finite presentations of monoids, allowing computations on
  those and their manipulation by a computer. Even better, when the presentation
  is confluent and terminating, they provide one with a notion of canonical
  representative of the elements of the presented monoid. Polygraphs are a
  higher-dimensional generalization of this notion of presentation, from the
  setting of monoids to the much more general setting of $n$-categories. One of
  the main purposes of this article is to give a progressive introduction to the
  notion of \emph{higher-dimensional rewriting system} provided by polygraphs,
  and describe its links with classical rewriting theory, string and term
  rewriting systems in particular. After introducing the general setting, we
  will be interested in proving local confluence for polygraphs presenting
  2\nbd{}categories and introduce a framework in which a finite 3-dimensional
  rewriting system admits a finite number of critical pairs.
\end{abstract}

\maketitle

Recent developments in category theory have established higher-dimensional
categories as a fundamental theoretical setting in order to study situations
arising in various areas of mathematics, physics and computer science. A nice
survey of these can be found in~\cite{baez:rosetta-stone}, explaining how the
use of category theory enables one to unify these apparently unrelated fields of
science, by revealing that their intrinsic algebraic structures are in fact
closely connected. In the last decade, higher-dimensional categories have
therefore emerged as a tool of everyday use for many scientists. The motivation
behind the concept of higher dimensions here is that, in order to have a
fine-grained understanding of the algebraic structures at stake, one should not
only consider morphisms between objects involved, but also morphisms between
morphisms (\ie 2-dimensional morphisms), morphisms between morphisms between
morphisms (\ie 3-dimensional morphisms), and so on. For example, the starting
point of algebraic topology~\cite{hatcher:algebraic-topology} is that one should
not consider points and paths between them in topological spaces, but also
homotopies between paths, and can be refined by also considering homotopies
between homotopies and so on.

The categorical structures considered nowadays are thus becoming more and more
complex, which enables them to capture many details, but the proofs are becoming
more and more complicated too, and we are facing the urge for new tools, both of
a theoretical and practical nature, in order to make them easier and more
manageable. In particular, many proofs of even conceptually simple facts involve
showing the commutativity of diagrams which are big both in size and dimension:
for example, the proof that a commutative strong monad on a symmetric monoidal
category (which is a particular 3-category) is the same notion as a symmetric
monoidal monad involves checking the commutativity of many large diagrams as
detailed in the appendix of~\cite{goubault:logical-relations}, as another
example, the very definition of tricategories (a weak form of 3-category) as
well as the associated morphisms takes up pages, not to mention the associated
coherence theorem~\cite{gordon-power-street:coh-tricat, gurski:tricat}. These
are, among many other, striking examples of the fact that we are slowly
approaching the size limits of diagrams that can be computed by hand in a decent
time, or decently written on paper. Moreover, these computations do not really
constitute the conceptually interesting part of the work, and are often
considered as ``routine checks'' because they are very repetitive and
systematic\ldots which is a good point from a computer scientist's point of
view: we have hope to be able to develop software to check and automate them!

\bigskip

With this goal in mind, rewriting theory is a natural candidate for providing
powerful tools in order to study these algebraic structures. Its uses have
proven very useful for example to construct and manipulate \emph{presentations}
of monoids, which are descriptions of monoids by generators and relations: a
string rewriting system is simply a presentation of a monoid where the relations
are oriented so that they form rewriting rules. When the rewriting system is
normalizing (which is in particular the case when it is locally confluent and
terminating), it provides one with a notion of canonical representative of words
modulo the relations, which turns out to be crucial to manipulate the
presentations. Starting from this point, Burroni has introduced the notion of
\emph{polygraph}~\cite{burroni:higher-word}, rediscovering Street and Power's
\emph{computads}~\cite{street:limit-indexed-by-functors, power:n-cat-pasting},
generalizing the notion of presentation from monoids to higher-dimensional
categories, thus providing us with a good notion of \emph{higher-dimensional
  rewriting system}, which is able to describe and manipulate free
$n$-categories modulo equations on $n$\nbd{}cells. This was later on used intensively
by Lafont in order to construct presentations of various monoidal
categories~\cite{lafont:boolean-circuits} and further studied by
Guiraud~\cite{guiraud:these, guiraud:termination-3-rewr,
  guiraud:three-dimensions-proofs, guiraud:presentations-petri-nets} and
Malbos~\cite{guiraud-malbos:higher-cat-fdt}.

However, much more than a mere adaptation of the well-known techniques of
rewriting theory is needed here. In particular, Lafont discovered a very
interesting fact: contrarily to string or term rewriting systems, these
generalized rewriting systems can give rise to an \emph{infinite} number of
critical pairs, even when they are \emph{finite}!
In this article, we will be interested in extending techniques for computing
critical pairs (we will not detail techniques for showing
termination~\cite{guiraud:termination-3-rewr}). Its main contribution is to
address this problem, in the case of 3-dimensional rewriting systems, by
generalizing the notion of critical pair in order to recover a finite number of
critical pairs for finite rewriting systems.
The present work constitutes a first step in the direction of generalizing
rewriting theory to higher dimensions, but a lot of efforts remain to be done in
order to achieve the task.

\bigskip

Unfortunately, the idea of higher-dimensional rewriting system does not seem to
have spread very much among the community of rewriting theory. We believe this
is partly due to the fact that its fundamental concepts are scattered in various
papers, which sometimes require a strong categorical background, or are even
sometimes considered as ``folklore''. In order to address this, we have done our
best to gradually introduce the topic, in a way accessible to people familiar
with rewriting theory~\cite{baader-nipkow:trat, terese:trs} and basic
(1\nbd{}dimensional) category theory~\cite{maclane:cwm}, and tried to write an
introductory article on the subject. We have favored the exposition instead of
handling directly the most technical matters, trying to provide intuitions about
the structures involved, and giving bibliographical references for the reader
interested in details concerning a particular point. Likewise, a completely
formal and more abstract presentation of all the contributions of the author in
the article -- as well as the other developments exposed by the author
in~\cite{mimram:critical-pairs} -- shall be presented in a subsequent
article~\cite{mimram:2-cp}. We feel that an introductory paper such as this one
was really necessary first, in order for the new developments to be accessible
to a broad audience.

\bigskip

We first recall in Section~\ref{sec:hdrs} the notion of higher-dimensional
rewriting system provided by polygraphs and explain its links with string and
term rewriting systems. The following Section~\ref{sec:pres-n-cat} describes how
these rewriting systems provide a good notion of presentation of
$n$-categories. Finally, we explain in Section~\ref{sec:repr-2-cp} why a finite
3-dimensional rewriting system can give rise to an infinite number of critical
pairs and introduce a theoretical setting to overcome this issue, and conclude
in Section~\ref{sec:future-work}.




\section{Towards higher-dimensional rewriting systems}
\label{sec:hdrs}
In order to generalize the notion of rewriting system to higher dimensions, one
should follow the principles established by topology and category theory: a
$0$-dimensional rewriting system should be a set of ``points'' and an
$(n+1)$\nbd{}dimensional rewriting system should consist of rules which
\emph{rewrite rewriting paths} in an $n$\nbd{}dimensional rewriting
system. Informally, one should think of the well-known situation arising in the
study of $\lambda$\nbd{}calculus: if we consider the $\lambda$\nbd{}terms modulo
$\alpha$-conversion as $0$\nbd{}dimensional terms, the rules for
$\beta$\nbd{}reduction would play the role of a 1\nbd{}dimensional rewriting
system, and the rules for standardization~\cite{mellies:ast}, which rewrite any
reduction sequence into a ``standard'' one, would play the role of a 2\nbd{}dimensional
rewriting system because they rewrite rewriting paths.

The idea is quite appealing, but it is not clear at first how the classical
rewriting frameworks should fit in this picture, and one should first try to
understand in a uniform way the two major examples of rewriting theory:
\emph{string rewriting systems} and \emph{term rewriting systems}. Namely,
string rewriting systems should intuitively be part of $2$-dimensional rewriting
systems because they rewrite words which can be considered as $1$-dimensional
objects: a word~$abc$ can be seen geometrically as a path
\[
\vxym{
  \ar[r]^a&\ar[r]^b&\ar[r]^c&\\
}
\]
Similarly, term rewriting systems should be part of $3$-dimensional rewriting
systems because they rewrite terms which can be considered as $2$-dimensional
objects, if we picture a term as a $2$\nbd{}dimensional diagram, using the
representation often used for sharing graphs: for example, the
term~$f(x,g(y,z))$ can be represented by
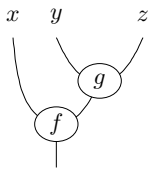
\[
\strid{dim_term}
\]
However, this intuition does not give precise directions about how a
4\nbd{}dimensional rewriting system should be defined for example. It turns out
that generalizing a bit the settings of string and term rewriting systems
reveals an inductive pattern in the definition of these rewriting frameworks,
whose discovery motivated the definition of \emph{polygraphs}, which can be
considered as the good notion of higher-dimensional rewriting system. We shall
begin by explaining this in details.

\subsection{Graphs and 1-dimensional rewriting systems}
\subsubsection{Definition}
Rewriting systems in dimension~$1$ should be defined in the most simple way. The
fact that a topological space of dimension 0 is a point suggests that we should
define a \emph{0-signature} simply as a set~$E_0$, whose elements are called
(0-dimensional) \emph{terms}. A 1-dimensional rewriting system on the
signature~$E_0$ is a set of \emph{rules} rewriting a term into another term. It
therefore consists of a set~$E_1$, whose elements are the \emph{rules}, together
with two functions~$s,t:E_1\to E_0$ which to every rule $r\in E_1$ associate its
\emph{source}~$s(r)$ and its \emph{target}~$t(r)$. In other words, a 1-rewriting
system is a diagram
\begin{equation}
  \label{eq:rs-1}
  \vxym{
    E_0&\ar@<-0.7ex>[l]_{s}\ar@<0.7ex>[l]^{t}E_1
  }
\end{equation}
in the category~$\Set$ (the category whose objects are sets and morphisms are
functions). We sometimes write
\[
r\qcolon A\qto B
\]
to indicate that~$r$ is a rule with~$A$ as source and~$B$ as target and say
that~$A$ \emph{one-step rewrites} by~$r$ to~$B$. We simply write~$A\to B$ when
there exists a rule $r:A\to B$, write~$\to^*$ for the transitive closure of the
relation~$\to$, and say that a term~$A$ \emph{rewrites} to a term~$B$
when~$A\to^*B$.

\subsubsection{The free category generated by a rewriting system}
\label{sec:rs-1-fcat}
Notice that a rewriting system of the form~\eqref{eq:rs-1} can also be seen as a
graph with~$E_0$ as set of vertices and~$E_1$ as set of edges, the functions~$s$
and~$t$ indicating the source and the target of an edge. Such a graph generates
a \emph{free category}~$\C$ which is the smallest category with~$E_0$ as set of
objects and containing all the edges~$r:A\to B$ in~$E_1$ as morphisms. We
write~$E_1^*$ for the set of morphisms of this category, and~$s^*,t^*:E_1^*\to
E_0$ for the functions which to every morphism of this category associate its
source and target respectively. Since the set~$E_1^*$ contains the set~$E_1$,
there is an injection~$i_1:E_1\to E_1^*$. Moreover, since~$\C$ is the free
category generated by the graph~\eqref{eq:rs-1}, $E_1^*$ is the smallest set
containing~$E_1$ and closed under identities and composition: $E_1^*$ is the set
of \emph{paths} in the graph, with concatenation as composition and empty paths
as identities, with the obvious source and target functions~$s_1^*$
and~$t_1^*$. The data generated by this free construction can be summarized by a
diagram
\[
\vxym{
  &\ar@<-0.7ex>[dl]_{s}\ar@<0.7ex>[dl]^{t}E_1\ar[d]^{i_1}\\
  E_0&\ar@<-0.7ex>[l]_-{\fpoly{s}}\ar@<0.7ex>[l]^-{\fpoly{t}}\fpoly{E_1}\\
}
\]
in~$\Set$, together with the composition morphism~$E_1^*\times_{E_0}E_1^*\to
E_1^*$ which to every pair of composable morphisms associate their composition
(\ie the concatenation of the two paths) and the identity morphism~$E_0\to
E_1^*$ which to every object in~$E_0$ associates an identity morphism (\ie the
empty path on this object): these two last operations define a structure of
category on the
graph~$\vxym{E_0&\ar@<-0.7ex>[l]_{\fpoly{s}}\ar@<0.7ex>[l]^{\fpoly{t}}\fpoly{E_1}}$
(a category is a graph together with composition and identity operations
satisfying axioms).
Moreover, the diagram above ``commutes'' in the sense that
\[
s^*\circ i=s
\qqtand
t^*\circ i=t
\]
which expresses the fact that if~$r:A\to B$ is an arrow in the graph, then it
has the same source and target if we see it as a path of length 1 in the
graph. Notice that, by definition, a term~$A$ rewrites to a term~$B$ when there
exists a path from~$A$ to~$B$ in the graph, \ie when there exists a morphism
from~$A$ to~$B$ in the generated category: the morphisms of the freely generated
category can be seen as rewriting sequences.

\begin{exa}
  \label{ex:fcat}
  Let us give examples of categories generated by graphs.
  \begin{enumerate}
  \item Consider the graph with two edges~$E_0=A,B$ and two arrows~$f:A\to B$
    and~\hbox{$g:B\to A$}
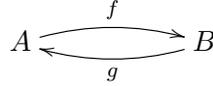
    \[
    \vxym{
      A\ar@/^/[rr]^f&&\ar@/^/[ll]^gB
    }
    \]
    \ie formally
    \[
    E_0=\{A,B\}
    \qquad
    E_1=\{f,g\}
    \qquad
    s(f)=t(g)=A
    \qquad
    t(f)=s(g)=B
    \]
    It generates a category which has the words of the form $(fg)^*$ as
    morphisms from~$A$ to~$A$ (where the star denotes the Kleene star),
    $(fg)^*f$ from~$A$ to $B$, $g(fg)^*$ from $B$ to $A$, and $(gf)^*$ from $B$
    to $B$, \ie formally $E_1^*=(fg)^*\uplus g(fg)^*\uplus (fg)^*f\uplus(gf)^*$.
  \item The following graph with one vertex and one arrow
    \[
    \vxym{
      \ast\ar@(ur,dr)^1
    }
    \]
    generates a category with one object isomorphic to the category whose
    morphisms are natural integers~$\N$ with addition as composition and $0$ as
    identity.
  \item More generally, given a set~$\Sigma$, the graph~$G_\Sigma$ with one
    object~$\ast$ and~$\Sigma$ as set of arrows (all from~$\ast$ to~$\ast$)
    generates a category isomorphic to the category with one object and the free
    monoid~$\Sigma^*$ on~$\Sigma$ as set of morphisms, with concatenation as
    composition and empty word as identity. In other words, the category
    generated by~$G_\Sigma$ is essentially the same as the free
    monoid~$\Sigma^*$ generated by~$\Sigma$.
  \end{enumerate}
\end{exa}

\subsection{String and 2-dimensional rewriting system}
\label{sec:rs-2}
\subsubsection{String rewriting systems}
Suppose given an \emph{alphabet}~$E_1$. We write~$E_1^*$ for the free monoid
generated by~$E_1$ and $i_1:E_1\to E_1^*$ for the injection sending a letter of
the alphabet to the corresponding word with only one letter. A \emph{string
  rewriting system}~$R$ on the alphabet~$E_1$ is usually defined as a relation
$R\subseteq E_1^*\times E_1^*$ over the free monoid generated by~$E_1$. One says
that a word $u$ \emph{one-step rewrites} to a word~$v$, what we write~$u\To v$,
when there exists two words~$w_1,w_2\in E_1^*$ such that
\begin{equation}
  \label{eq:rs-2-osr}
  u=w_1u'w_2
  \qquad
  v=w_1v'w_2
  \qtand
  (u',v')\in R
\end{equation}
\noindent If we write~$\To^*$ for the reflexive and transitive closure of the
relation~$\To$, a word~$u$ \emph{rewrites} to a word~$v$ when~$u\To^* v$.

In order to make the definition closer to the one given in previous section, we
can make a small generalization and allow to have two distinct rules both
rewriting the same word~$u$ to the same word~$v$, by giving \emph{names} to the
rules. Such a string rewriting system can be formalized as a set~$E_2$ whose
elements are the (names of) the rules together with two
functions~$s_1,t_1:E_2\to E_1^*$ indicating respectively the \emph{source} and
\emph{target} of a rule:
\[
\vxym{
  E_1\ar[d]_{i_1}&E_2\ar@<-0.7ex>[dl]_{s_1}\ar@<0.7ex>[dl]^{t_1}\\
  E_1^*
}
\]
and again we write $\rho:u\To v$ to indicate that~$\rho\in E_2$ is a rule
with~$s_1(\rho)=u$ and~$t_1(\rho)=v$.

\begin{exa}
  \label{ex:str-rs}
  One can for example consider the rewriting system on an alphabet with two
  letters~$a,b$ and a rule~$ba\To ab$, which would be formalized by
  \[
  E_1=\{a,b\}
  \qquad
  E_2=\{\rho\}
  \qquad
  s_1(\rho)=b\otimes a
  \qquad
  t_1(\rho)=a\otimes b
  \]
  where $\otimes$ denotes the concatenation operation on words.
\end{exa}

\noindent
More details and examples of uses of such rewriting systems are given in
Section~\ref{sec:mon-pres}.

\subsubsection{2-dimensional rewriting systems}
\label{sec:2-rs}
Remember that our general plan is to express 2\nbd{}dimensional rewriting
systems as rewriting on 1-dimensional rewriting paths. In order to do so, one
should have in mind that a set is ``the same as'' a graph with one
vertex. Formally, there exist two functors
\[
\vxym{
  \Set\ar@/^/[rr]^F&&\ar@/^/[ll]^G\Graph_*
}
\]
which are inverse one of each other, where~$\Graph_*$ is the full subcategory
of~$\Graph$ (the category of graphs) whose objects are graphs with only one
vertex: the functor~$F$ sends a set~$\Sigma$ to the graph with one vertex~$\ast$
and~$\Sigma$ as set of arrows (all going from~$\ast$ to~$\ast$) and the
functor~$G$ sends a graph with one vertex to its set of edges. Similarly, a
monoid is ``the same as'' a category with one object, the multiplication and
unit of a monoid corresponding to the composition and identity of the
category. For instance, we have explained in the last point of
Example~\ref{ex:fcat} that the free monoid generated by a set is, with respect
to this point of view, ``the same as'' the free category generated by the graph
corresponding to the set.

Therefore, instead of having a mere set as alphabet, one might equivalently
declare that an alphabet is a 1-dimensional rewriting system
\begin{equation}
  \label{eq:rs-2-alph}
  \vxym{
    E_0&\ar@<-0.7ex>[l]_{s_0}\ar@<0.7ex>[l]^{t_0}E_1
  }
\end{equation}
as defined in Section~\ref{sec:rs-2}, with the set~$E_0$ being reduced to one
element. Again, this does not bring more information than the set~$E_1$ itself,
since the set~$E_0$ is uniquely defined up to a canonical isomorphism (it
contains only one object) and the functions~$s_0,t_0:E_1\to E_0$ are uniquely
defined because~$E_0$ is terminal in~$\Set$. If we write~$E_1^*$ for the set of
morphisms of the category generated by the graph~\eqref{eq:rs-2-alph} as in
previous section, a 2-dimensional rewriting system can thus be defined as a
diagram
\begin{equation}
  \label{eq:rs-2}
  \vxym{
    &E_1\ar[d]_{i_1}\ar@<-0.7ex>[dl]_{s_0}\ar@<0.7ex>[dl]^{t_0}&E_2\ar@<-0.7ex>[dl]_{s_1}\ar@<0.7ex>[dl]^{t_1}\\
    E_0&\ar@<-0.7ex>[l]_{\fpoly{s_0}}\ar@<0.7ex>[l]^{\fpoly{t_0}}\fpoly{E_1}\\
  }
\end{equation}
in~$\Set$, with~$E_0$ containing only one element, together with a structure of
category on the graph
\begin{equation}
  \label{eq:rs-2-gcat}
  \vxym{
    E_0&\ar@<-0.7ex>[l]_{\fpoly{s_0}}\ar@<0.7ex>[l]^{\fpoly{t_0}}\fpoly{E_1}
  }
\end{equation}

It turns out that the supposition that the set~$E_0$ is reduced to one element
is not necessary in order to proceed with the usual developments of string
rewriting theory. We thus drop this condition in the following and suppose that
an alphabet for a 2-dimensional rewriting system is any 1-dimensional rewriting
system. This amounts to generalizing string rewriting systems to the case where
the letters $a:A\to B$ are typed (with~$A$ as source and~$B$ as target), and to
consider only \emph{composable} string of letters as words (the words are thus
also typed). In order for substitution to be well-defined, we have to suppose
that a word $u:A\to B$ can only be rewritten to a word~$v:A'\to B'$ of the same
type (\ie $A'=A$ and $B'=B$). It is enough to suppose that this property is
satisfied for the rules, which amounts to impose the further axiom that in any
2-dimensional rewriting system~\eqref{eq:rs-2}, the equations
\begin{equation*}
  \label{eq:gen-1-globular}
  s_0^*\circ s_1
  =
  s_0^*\circ t_1
  \qqtand
  t_0^*\circ s_1
  =
  t_0^*\circ t_1
\end{equation*}
are satisfied (these are sometimes called the \emph{globular equations}).
An example of such a rewriting system is given in Section~\ref{sec:rs-2-ex}.

\subsubsection{The free 2-category generated by a rewriting system}
\label{sec:rs-free-2-cat}
In the same way that a 1\nbd{}dimensional rewriting system freely generates a
1-category (\ie a category), a 2-dimensional rewriting system freely generates a
2-category: it is the smallest 2-category with~\eqref{eq:rs-2-gcat} as
underlying category and the containing the elements of~$E_2$ as 2-cells. We
recall that the notion of 2-category is defined as follows.

\begin{defi}
  \label{def:2-cat}
  A \emph{2-category}~$\C$ consists of the following data.
  \begin{itemize}
  \item A class~$\C_0$ of \emph{0-cells}.
  \item A category~$\C(A,B)$ for every pair of 0-cells~$A$ and~$B$. Its objects
    $f:A\to B$ are called \emph{1-cells}, its morphisms $\alpha:f\To g:A\to B$
    are called \emph{2-cells}, composition is written~$\circ$ and called
    \emph{vertical composition}, and identities are called \emph{vertical
      identities}.
  \item A functor $\otimes_{A,B,C}:\C(A,B)\times\C(B,C)\to\C(A,C)$ for every
    objects $A$, $B$ and $C$ called \emph{horizontal composition} (we will drop
    the subscripts of $\otimes$ in the following).
  \item A 1-cell $\id_A:A\to A$ for every object~$A$ called \emph{vertical
      identity}.
  \end{itemize}
  These should be such that the following properties are satisfied.
  \begin{itemize}
  \item Horizontal composition is associative: for every 0-cells $A$, $B$, $C$
    and $D$, for every 1-cells $f,f':A\to B$, $g,g':B\to C$ and $h,h':C\to D$,
    for every 2-cells $\alpha:f\To f'$, $\beta:g\To g'$ and $\gamma:h\To h'$,
    \[
    (f\otimes g)\otimes h=f\otimes(g\otimes h)
    \quad
    (\alpha\otimes\beta)\otimes\gamma=\alpha\otimes(\beta\otimes\gamma)
    \quad
    (f'\otimes g')\otimes h'=f'\otimes(g'\otimes h')
    \]
  \item Horizontal identities are neutral elements for horizontal composition:
    for every 0\nbd{}cells $A$ and $B$, for every 1-cells $f,f':A\to B$, for
    every 2-cell $\alpha:f\To f'$,
    \[
    \id_A\otimes f=f=f\otimes\id_B
    \quad
    \id_{\id_A}\otimes\alpha=\alpha=\alpha\otimes\id_{\id_B}
    \quad
    \id_A\otimes f'=f'=f'\otimes\id_B
    \]
  \end{itemize}
\end{defi}


If we write~$E_2^*$ for the set of 2-cells of the 2-category generated by the
2-rewriting system~\eqref{eq:rs-2}, $i_2:E_2\to E_2^*$ for the canonical
injection expressing the inclusion of~$E_2$ into~$E_2^*$ and
$s_1^*,t_1^*:E_2\to E_1^*$, we thus get a diagram
\[
\vxym{
  &E_1\ar[d]_{i_1}\ar@<-0.7ex>[dl]_{s_0}\ar@<0.7ex>[dl]^{t_0}&E_2\ar@<-0.7ex>[dl]_{s_1}\ar@<0.7ex>[dl]^{t_1}\ar[d]^{i_2}\\
  E_0&\ar@<-0.7ex>[l]_{\fpoly{s_0}}\ar@<0.7ex>[l]^{\fpoly{t_0}}\fpoly{E_1}&\ar@<-0.7ex>[l]_{\fpoly{s_1}}\ar@<0.7ex>[l]^{\fpoly{t_1}}\fpoly{E_2}\\
}
\]
in~$\Set$, together with a structure of 2-category on the 2-graph
\[
\vxym{
  E_0&\ar@<-0.7ex>[l]_{\fpoly{s_0}}\ar@<0.7ex>[l]^{\fpoly{t_0}}\fpoly{E_1}&\ar@<-0.7ex>[l]_{\fpoly{s_1}}\ar@<0.7ex>[l]^{\fpoly{t_1}}\fpoly{E_2}
}
\]
which commutes in the sense that
\[
s_1^*\circ i_2=s_1
\qqtand
t_1^*\circ i_2=t_1
\]

In the same way as previously, the 2-cells of this 2-category correspond to the
rewriting paths: a (typed) word~$u:A\to B$ rewrites to a (typed) word~$w:A\to B$
if and only if there exists a 2-cell~$\alpha:u\To v:A\to B$. In particular, if
there is only one 0-cell, this notion of rewriting corresponds the usual notion
of rewriting for string rewriting systems. As previously, notice that a 1-cell
is rewritten into a \emph{parallel} 1-cell, \ie a 1-cell having the same source
and same target.


\subsubsection{A diagrammatic notation for 2-cells}
\label{sec:2-cells}
Before going on describing how to incorporate string rewriting systems in our
framework, we shall try to understand a bit more 2-categories, and the previous
construction.

Given a 2-category~$\C$, a 2-cell $\alpha:u\To v:A\to B$ is often written
diagrammatically as
\[
\vxym{
A\ar@/^2ex/[r]^u\ar@/_2ex/[r]_v\ar@{}[r]|-{\alpha\Downarrow}&B
}
\]
With this notation, vertical composition corresponds to vertical juxtaposition:
the composite~$\beta\circ\alpha:u\To w:A\to B$ of a pair of 2-cells $\alpha:u\To
v:A\to B$ and $\beta:v\To w:A\to B$ is denoted by
\begin{equation}
  \label{eq:vcomp}
  \vxym{
    A\ar@/^4ex/[r]^u\ar[r]|-v\ar@/_4ex/[r]_w\ar@{}[r]^-{\alpha\Downarrow}\ar@{}[r]_-{\beta\Downarrow}&B
  }
\end{equation}
and the horizontal composition~$\alpha\otimes\beta:(u\otimes v)\To(u'\otimes
v'):A\to C$ of a pair of 2-cells $\alpha:u\To u':A\to B$ and~$\beta:v\To
v':B\to C$ corresponds diagrammatically to horizontal juxtaposition
\begin{equation}
  \label{eq:hcomp}
  \vxym{
    A\ar@/^2ex/[r]^u\ar@/_2ex/[r]_{u'}\ar@{}[r]|-{\alpha\Downarrow}&B\ar@/^2ex/[r]^v\ar@/_2ex/[r]_{v'}\ar@{}[r]|-{\beta\Downarrow}&C
  }
\end{equation}
For every 1-cell~$u:A\to B$, a vertical identity 2-cell
\[
\vxym{
  A\ar@/^2ex/[r]^u\ar@/_2ex/[r]_u\ar@{}[r]|-{\id_f\Downarrow}&B
}
\]
should be given which acts as neutral element for vertical composition.

The functoriality of the horizontal composition implies a last fundamental
identity called the \emph{exchange law}. Given four 2-cells $\alpha$, $\beta$,
$\alpha'$ and~$\beta'$ as in the diagram
\begin{equation}
  \label{eq:exch-law}
  \vxym{
    A\ar@/^4ex/[r]^u\ar[r]|-v\ar@/_4ex/[r]_w\ar@{}[r]^-{\alpha\Downarrow}\ar@{}[r]_-{\beta\Downarrow}&B\ar@/^4ex/[r]^{u'}\ar[r]|-{v'}\ar@/_4ex/[r]_{w'}\ar@{}[r]^-{\alpha'\Downarrow}\ar@{}[r]_-{\beta'\Downarrow}&C
  }
\end{equation}
there are two ways to compose them altogether: either first vertically and then
horizontally or the converse. The exchange law which is satisfied in any
2-category states that these two composites are equal:
\[
(\beta\circ\alpha)\otimes(\beta'\circ\alpha')
\qeq
(\beta\otimes\beta')\circ(\alpha\otimes\alpha')
\]
Thanks to this, the diagram~\eqref{eq:exch-law} is not ambiguous. This law can
equivalently be reformulated slightly differently as the \emph{Godement law}:
given two 2-cells $\alpha$ and~$\beta$ as in~\eqref{eq:hcomp}, the following
identity is always satisfied:
\begin{equation}
  \label{eq:godement}
  (\id_{u'}\otimes\beta)\circ(\alpha\otimes\id_v)
  \qeq
  (\alpha\otimes\id_{v'})\circ(\id_u\otimes\beta)
\end{equation}

\subsubsection{An example}
\label{sec:rs-2-ex}
Consider a 2-dimensional rewriting system of the form~\eqref{eq:rs-2} with
\[
E_0=\{A,B\}
\qquad
E_1=\{f,g\}
\qquad
E_2=\{\rho\}
\]
and source and target map defined by
\[
s_0(f)=t_0(g)=A
\qquad
t_0(f)=s_0(g)=B
\qquad
s_1(\rho)=f\otimes g\otimes f\otimes g
\qquad
t_1(\rho)=f\otimes g
\]
\ie using the previously introduced notations
\[
f:A\to B
\qquad
g:B\to A
\qquad
\rho:f\otimes g\otimes f\otimes g\To f\otimes g:A\to B
\]
In particular, the 1-rewriting system playing the role of signature for this
2-rewriting system is the one given in first point of Example~\ref{ex:fcat}, and
the~$\otimes$ operation denotes the composition of the category it generates.

As explained earlier, $A$ and~$B$ should be seen as types for words, $f$ and~$g$
should be seen as the letters of the signature and the $\otimes$ operation as
the concatenation for words (we thus sometimes omit it and write~$fg$ instead
of~$f\otimes g$). These letters are typed and one only considers words obtained
by composing compatible letters; for example~$f\otimes g\otimes f$ is a word of
type~$A\to B$, graphically represented as the path
\[
\vxym{
  A\ar[r]^f&B\ar[r]^g&A\ar[r]^f&B
}
\]
whereas $f\otimes f$ does not make sense. Finally, $\rho$ is a rewriting rule,
which we may represent as a 2-cell
\begin{equation}
  \label{eq:ex-rho}
  \vxym{
    &B\ar[r]^g&A\ar@{}[d]|-{\rho\Downarrow}\ar[r]^f&B\ar@/^/[dr]^g&\\
    A\ar@/^/[ur]^f\ar[rr]_f&&B\ar[rr]_g&&A\\
  }
\end{equation}
which transforms a word into another.

The rewriting paths are the morphisms in the 2-category generated by the
2-rewriting system: its 2-cells are formal vertical and horizontal composites of
rules quotiented by the axioms of 2-categories. Vertical compositions correspond
to sequences of rewriting: if a word~$u$ rewrites to a word~$v$, which is
witnessed by a morphism~$\alpha:u\To v$, and the word~$v$ rewrites to a
word~$w$, which is witnessed by a morphism~$\beta:v\To w$, then the composite
morphism $\beta\circ\alpha:u\To w$ witnesses the fact that~$u$ rewrites to~$w$,
which is represented diagrammatically in~\eqref{eq:vcomp}. Horizontal
compositions correspond similarly to rewritings done in parallel: if a word~$u$
rewrites to~$u'$ by~$\alpha:u\To u'$ and a word~$v$ rewrites to~$v'$
by~$\beta:v\To v'$ then $u\otimes v$ rewrites to~$u'\otimes v'$
by~$\alpha\otimes\beta$, as shown in~\eqref{eq:hcomp}. For example, the usual
condition for one-step rewrites~\eqref{eq:rs-2-osr} can be recovered as follows:
the generator~$\rho$ rewrites~$fgfg$ into~$fg$, so, given two words~$u:C\to A$
and~$v:A\to D$ (where~$C$ and~$D$ are either~$A$ or~$B$), the word~$ufgfgv$
rewrites to~$ufgv$, which is witnessed by the morphism
\[
\id_u\otimes\rho\otimes\id_v
\qcolon
u\otimes f\otimes g\otimes f\otimes g\otimes v
\qTo
u\otimes f\otimes g\otimes v
\qcolon
C\qto D
\]
which can be represented as
\[
\vxym{
  &&B\ar[r]^g&A\ar@{}[dd]|-{\rho\Downarrow}\ar[r]^f&B\ar@/^/[dr]^g&\\
  C\ar@/^4ex/[r]^u\ar@/_4ex/[r]_u\ar@{}[r]|-{\id_u\Downarrow}&A\ar@/^/[ur]^f\ar@/_/[drr]_f&&&&A\ar@/^4ex/[r]^v\ar@/_4ex/[r]_v\ar@{}[r]|-{\id_v\Downarrow}&D\\
  &&&B\ar@/_/[urr]_g\\
}
\]
We usually omit drawing the identity 2-cells and simply write
\[
\vxym{
  &&B\ar[r]^g&A\ar@{}[d]|-{\rho\Downarrow}\ar[r]^f&B\ar@/^/[dr]^g&\\
  C\ar[r]^u&A\ar@/^/[ur]^f\ar[rr]_f&&B\ar[rr]_g&&A\ar[r]^v&D\\
}
\]
for this 2-cell.

As another example of the fact that 2-cells represent rewriting paths, consider
the word $fgfgfg$ which can be rewritten to~$fg$:
\[
fgfgfg
\qTo
fgfg
\qTo
fg
\]
Actually, there are \emph{two} such rewriting paths because, in the first step,
the rule~$\rho$ can be applied either on the left or on the right of the
word. These two distinct paths correspond to the morphisms
\[
\rho\circ(\rho\otimes f\otimes g)
\qqtand
\rho\circ(f\otimes g\otimes\rho)
\]
respectively represented diagrammatically by
\begin{equation}
  \label{eq:rs-2-ex-rew-1}
  \vxym{
    &B\ar[r]^g&A\ar@{}[d]|-{\rho\Downarrow}\ar[r]^f&B\ar@/^/[dr]^g&\\
    A\ar@/_/[drrr]_f\ar@/^/[ur]^f\ar[rr]_f&&B\ar[rr]_g&&A\ar[r]^f&B\ar[r]^g&A\\
    &&&B\ar@/_/[urrr]_g\ar@{}[u]|-{\rho\Downarrow}
  }
\end{equation}
and
\begin{equation*}
  \vxym{
    &&&B\ar[r]^g&A\ar@{}[d]|-{\rho\Downarrow}\ar[r]^f&B\ar@/^/[dr]^g&\\
    A\ar@/_/[drrr]_f\ar[r]^f&B\ar[r]^g&A\ar@/^/[ur]^f\ar[rr]_f&&B\ar[rr]_g&&A\\
    &&&B\ar@{}[u]|-{\rho\Downarrow}\ar@/_/[urrr]_g
  }
\end{equation*}
Notice that it is also possible in this setting to express the rewriting of two
disjoint parts of a word in parallel, such as
\[
\rho\otimes\rho
\qcolon
fgfgfgfg
\qTo
fgfg
\]
which is represented diagrammatically by
\[
\vxym{
  &B\ar[r]^g&A\ar[r]^f\ar@{}[d]|-{\rho\Downarrow}&B\ar@/^/[dr]^g&&B\ar[r]^g&A\ar[r]^f\ar@{}[d]|-{\rho\Downarrow}&B\ar@/^/[dr]^g&\\
  A\ar[rr]_f\ar@/^/[ur]^f&&B\ar[rr]_g&&A\ar[rr]_f\ar@/^/[ur]^f&&B\ar[rr]_g&&A\\
}
\]


\subsubsection{String diagrams}
The diagrammatic notation for cells in 2-categories introduced in
Section~\ref{sec:2-cells} is quite useful and convenient. However, from such diagrams,
it is sometimes difficult to get intuitions about the way morphisms should be
manipulated and having another, more ``geometrical'', representation of
morphisms can be very useful. Such a notation was formally introduced by Joyal
and Street~\cite{joyal-street:geometry-tensor-calculus} and is called
\emph{string diagrams}.

We shall first explain this notation on an example. Consider the 2-cell~$\rho$
in the 2\nbd{}category described in previous section, whose diagrammatic
representation is given in~\eqref{eq:ex-rho}. We can consider this 2-cell as a
device with four typed inputs respectively of type $f$, $g$, $f$ and $g$, which
produces two outputs respectively of type~$f$ and~$g$. It is thus natural to
depict it as
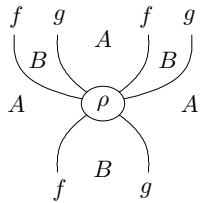
\[
\strid{ex_rho}
\]
The 2-cell~$\rho$ is represented by the point in the center, each of the 1-cells
in its input (\resp output) correspond to a wire going in (\resp out of) it. The
portions of the plane delimited by the wires correspond to 0-cells at the source
or target of the 1-cells corresponding to the delimiting wires. Since, in this
representation, 2-cells consist of points linked by strings (or wires), it is
called a \emph{string diagram}.


Vertical composition amounts to vertically juxtaposing diagrams and ``linking
the wires'', whereas horizontal composition is given by horizontal juxtaposition
of diagrams. The two diagrams for composition~\eqref{eq:vcomp}
and~\eqref{eq:hcomp} are thus drawn respectively by
\[
\strid{vcomp}
\qqtand
\strid{hcomp}
\]
and vertical identities are represented by wires. For example, the
string-diagrammatic notation of the morphism~\eqref{eq:rs-2-ex-rew-1} is
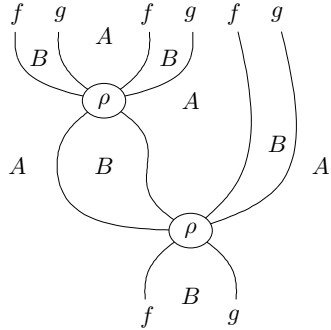
\[
\strid{ex_rho_rho}
\]

These diagrams should be considered modulo planar continuous deformations. In
particular, the Godement exchange law~\eqref{eq:godement} is satisfied
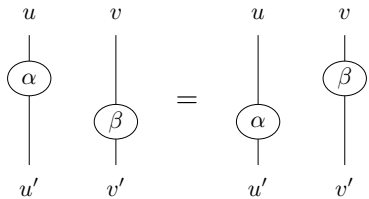
\[
\strid{godement_l}
\qeq
\strid{godement_r}
\]
because the diagram on the left can be deformed into the diagram on the right,
and more generally, it can be checked that all the axioms of 2-categories are
satisfied.

\subsection{Term and 3-dimensional rewriting systems}
\subsubsection{Term rewriting systems}
\label{sec:trs}
In the same way string rewriting systems can be generalized to 2-dimensional
rewriting systems, which rewrite rewriting paths in 1-dimensional rewriting
systems, term rewriting systems can be generalized to 3-dimensional rewriting
systems which rewrite rewriting paths in 2-dimensional rewriting systems.

A \emph{signature} $(\Sigma,a)$ consists of a set~$\Sigma$ called
\emph{alphabet}, whose elements are called \emph{symbols}, together with a
function $a:\Sigma\to\N$ which to every symbol~$f\in\Sigma$ associates its
\emph{arity}. We also suppose fixed a set~$\Var$ of \emph{variables}. A
\emph{term} is either a variable $x\in\Var$ or of the form $f(t_1,\ldots,t_n)$
where the $t_i$ are terms and~$f$ is a symbol of arity~$n$. Given two terms $t$
and $u$ and a variable~$x$, we denote by~$t[u/x]$ the term obtained from~$t$ by
substituting all the occurrences of the variable~$x$ by~$u$. We write~$\Sigma^*$
for the set of terms generated by the signature~$\Sigma$. A
\emph{substitution}~$\sigma:\Var\to\Sigma^*$ is a function which to every
variable associates a term. Given a term~$t$, we write~$t[\sigma]$ for the term
obtained from~$t$ by replacing every variable~$x$ by~$\sigma(x)$. A \emph{term
  rewriting system}~$R$ on such a signature is a set
$R\subseteq\Sigma^*\times\Sigma^*$ of pairs of terms. A term~$t$ \emph{one-step
  rewrites} to a term~$u$, what we write~$t\TO u$, when there exists a
rule~$(l,r)\in R$ and a substitution~$\sigma$ such that $u=l[\sigma]$
and~$v=r[\sigma]$. We write $\TO^*$ for the transitive closure of the
relation~$\TO$ on terms, and say that a term~$t$ \emph{rewrites} to a term~$u$
whenever~$t\TO^*u$.

We now generalize the notion of term rewriting system so that it appears as an
extension of 2-dimensional rewriting systems. Having done that, we shall be able
to see a term as a rewriting path in a 2-dimensional rewriting system: a
3-dimensional rewriting system will thus consist of rewriting rules which relate
terms, \ie rewriting rules which rewrite rewriting paths in a 2-dimensional
rewriting system. We should follow the intuition given by string diagrams. For
example, if we consider a signature~$\Sigma=\{f,g\}$, the two symbols~$f$
and~$g$ being of arity 2, as mentioned in the introduction, the
term~$f(x,g(y,z))$, where $x$, $y$ and $z$ are variables, can be drawn as a tree
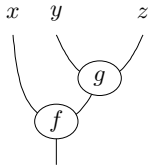
\[
\strid{dim_term}
\]
which looks very much like the string-diagrammatic notation for a morphism
$f\circ(\id\otimes g)$, whose diagrammatic notation would be
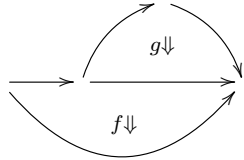
\[
\vxym{
  &&\ar@/^/[dr]\ar@{}[d]|-{g\Downarrow}\\
  \ar[r]\ar@/_6ex/[rrr]&\ar[rr]\ar@/^/[ur]\ar@{}[dr]|-{f\Downarrow}&&\\
  &&&\\
}
\]
In order to make this intuition formal, recall from Example~\ref{ex:fcat} that
the monoid~$\N$ of integers can be seen as the category generated by the graph
with only one vertex and one edge: if we write~$E_0=\{*\}$ for the set of
vertices ($*$ being the only vertex) and~$E_1=\{1\}$ for the set of edges ($1$
being the only edge, with~$*$ both as source and target) and follow the
notations of Section~\ref{sec:rs-1-fcat}, we have the diagram
\[
\vxym{
  &E_1\ar[d]_{i_1}\ar@<-0.7ex>[dl]_{s_0}\ar@<0.7ex>[dl]^{t_0}\\
  E_0&\ar@<-0.7ex>[l]_{\fpoly{s_0}}\ar@<0.7ex>[l]^{\fpoly{t_0}}\fpoly{E_1}\\
}
\]
where~$E_1^*\cong\N$ is the set of morphisms in the category generated by the
graph (\ie the paths in the graph). For example, the integer 3 corresponds to
the path $1\otimes 1\otimes 1$:
\[
\vxym{
  \ast\ar[r]^1&\ast\ar[r]^1&\ast\ar[r]^1&\ast
}
\]
and more generally an integer~$n$ to the composition of $n$~copies of~$1$. The
arity can thus be seen as an arrow~$a:\Sigma\to E_1^*$. In order to be
consistent with the previous notations, we write~$E_2$ instead of~$\Sigma$
and~$s_1$ instead of~$a$, since the arity denotes the number of inputs. We also
write~$t_1:E_2\to E_2^*$ for the constant function equal to~$1$: it denotes the
\emph{coarity} of a symbol (its number of outputs), which is always~$1$ in the
setting of term rewriting systems. The signature of a term rewriting system can
thus be specified by a 2-rewriting system
\[
\vxym{
  &E_1\ar[d]_{i_1}\ar@<-0.7ex>[dl]_{s_0}\ar@<0.7ex>[dl]^{t_0}&E_2\ar@<-0.7ex>[dl]_{s_1}\ar@<0.7ex>[dl]^{t_1}\\
  E_0&\ar@<-0.7ex>[l]_{\fpoly{s_0}}\ar@<0.7ex>[l]^{\fpoly{t_0}}\fpoly{E_1}\\
}
\]
where~$E_0=\{*\}$, $E_1=\{1\}$ and~$t_1$ is the constant function equal
to~$1$. For example, the previously considered signature with two symbols $f$,
$g$ of arity 2 is represented by the 2-rewriting system such that
\begin{equation}
  \label{eq:ex-rs-3-sig}
  E_2=\{f,g\}
  \qtand
  s_1(f)=s_1(g)=1\otimes 1
\end{equation}

Any such signature being a 2-rewriting system, it generates a 2-category~$\C$
with~$E_2^*$ as set of 2-cells. As previously, we write~$i_2:E_2\to E_2^*$ for
the canonical injection and \hbox{$s_1^*,t_1^*:E_2^*\to E_1^*$} for the maps
associating to each 2-cell its source and target respectively. Since~$E_0=\{*\}$,
the 2-category~$\C$ has only one 0-cell~$*$. Since~$E_1=\{1\}$, the 1-cells of
the 2-category~$\C$ are in bijection with~$\N$ (as explained above). It is easy
to show that the 2\nbd{}cells~$\alpha:n\To 1:*\to *$ of~$\C$ are in bijection
with linear terms with variables in~$x_1,\ldots,x_n$: a term is \emph{linear}
when all the variables~$x_1,\ldots,x_n$ appear exactly once in the term, in this
order, \eg with the signature~\eqref{eq:ex-rs-3-sig}, the term
$f(x_1,g(x_2,x_3))$ is linear whereas~$f(x_1,x_1)$ and~$g(x_2,x_1)$ are
not. More generally, the 2-cells $\alpha:n\To m:*\to *$ are in bijection with
linear $m$-uples of terms with~$n$ variables $x_1,\ldots,x_n$. For example, with
the signature~\eqref{eq:ex-rs-3-sig}, the term
\[
(f\circ(f\otimes\id_1))\otimes(g\circ(f\otimes f))
\]
which can be represented diagrammatically as
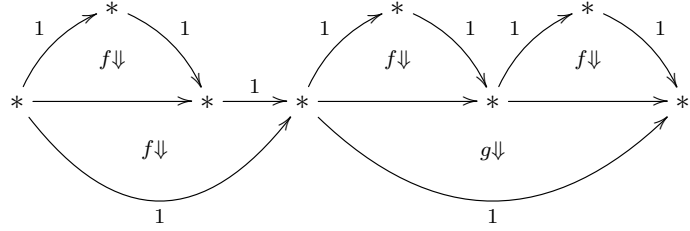
\[
\vxym{
  &\ast\ar@/^/[dr]^1\ar@{}[d]|-{f\Downarrow}&&&\ast\ar@/^/[dr]^1\ar@{}[d]|-{f\Downarrow}&&\ast\ar@/^/[dr]^1\ar@{}[d]|-{f\Downarrow}\\
  \ast\ar@/^/[ur]^1\ar[rr]\ar@/_8ex/[rrr]_1&&\ar@{}[dl]|-{f\Downarrow}\ast\ar[r]^1&\ast\ar@/^/[ur]^1\ar[rr]\ar@/_8ex/[rrrr]_1&&\ast\ar@/^/[ur]^1\ar[rr]\ar@{}[d]|-{g\Downarrow}&&\ast\\
  &&&&&&\\
}
\]
or string-diagrammatically as
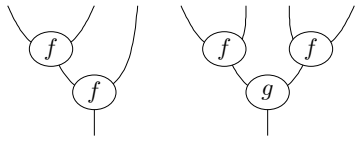
\[
\strid{lterm_ex}
\]
represents the pair $f(f(x_1,x_2),x_3),\ g(f(x_4,x_5),f(x_6,x_7))$.

A term rewriting system operating on linear terms can thus be thought as
rewriting rewriting paths in a 2-dimensional rewriting system and be represented
as a diagram
\begin{equation}
  \label{eq:rs-3}
  \vxym{
    &E_1\ar[d]_{i_1}\ar@<-0.7ex>[dl]_{s_0}\ar@<0.7ex>[dl]^{t_0}&E_2\ar[d]_{i_2}\ar@<-0.7ex>[dl]_{s_1}\ar@<0.7ex>[dl]^{t_1}&E_3\ar@<-0.7ex>[dl]_{s_2}\ar@<0.7ex>[dl]^{t_2}\\
    E_0&\ar@<-0.7ex>[l]_{\fpoly{s_0}}\ar@<0.7ex>[l]^{\fpoly{t_0}}\fpoly{E_1}&\ar@<-0.7ex>[l]_{\fpoly{s_1}}\ar@<0.7ex>[l]^{\fpoly{t_1}}\fpoly{E_2}\\
  }
\end{equation}
where~$E_2$ is the alphabet and~$E_3$ is the set of rewriting rules, such that
rules rewrite a term into a term of the same type:
\[
s_1^*\circ s_2=s_1^*\circ t_2
\qqtand
t_1^*\circ s_2=t_1^*\circ t_2
\]
and the three conditions below are satisfied:
\begin{enumerate}
\item the set~$E_0$ is reduced to one element~$*$,
\item the set~$E_1$ is reduced to one element~$1$,
\item the function~$t_2$ is the constant function equal to~$1$.
\end{enumerate}
In the following, we drop those three conditions in the general definition of
3-dimensional rewriting systems. Dropping condition (3) allows us to consider
terms with multiple outputs (whereas terms in term-rewriting systems always have
exactly one output). In particular, we will see in Section~\ref{sec:pres-law}
that this enables us to recover the usual setting of (non-linear) term rewriting
systems: non-linearity is obtained by adding to the signature symbols to
explicitly duplicate, erase and swap variables. Dropping condition (2) allows us
to consider typed terms. For example, a symbol~$f\in E_2$ with~$s_1(f)=A\otimes
B$ and~$t_1(f)=C\otimes D$, where $A,B,C,D\in E_1$ means that~$f$ has two inputs
of respective types~$A$ and~$B$ and two outputs of respective types~$C$ and~$D$;
diagrammatically,
\[
\vxym{
&\ar@/^/[dr]^B\\
\ar@/^/[ur]^A\ar@/_/[dr]_C&f\Downarrow&\\
&\ar@/_/[ur]_D\\
}
\qquad\qquad\qquad\qquad\qquad
\strid{f_abcd}
\]
Finally, dropping condition (1) allows to have typed types, as explained in
Section~\ref{sec:rs-2}.

Notice that usually, term rewriting systems are considered to generalize string
rewriting system because a letter can be seen as a symbol of arity~$1$. While
this is true, this is not the right point of view! At the light of the previous
explanations, we see that term rewriting systems do generalize string rewriting
systems in the sense that they rewrite rewriting paths in string rewriting systems.

\subsection{Polygraphs}
\label{sec:polygraphs}
We believe that the general pattern should have become clear now, and one can
generalize the situation in order to inductively define a notion of
$n$-dimensional rewriting system. This was formalized by Burroni under the name
of~\emph{$n$-polygraph}~\cite{burroni:higher-word}. The article introducing
polygraphs was in fact rediscovering the notion of \emph{$n$-computad} invented
17 years earlier by Street in its 2\nbd{}dimensional
version~\cite{street:limit-indexed-by-functors} and later on generalized to
higher dimensions by Power~\cite{power:n-cat-pasting}. We insist here on the
contribution of Burroni, who first found and formalized connections between
these general constructions and rewriting theory.

The notion of $n$\nbd{}polygraph is defined inductively (on~$n\in\N$) as
follows. A $0$-polygraph is a set~$E_0$. Now suppose given an $n$-polygraph, \ie
a diagram
\begin{equation}
  \label{eq:n-pol}
  \vxym{
    E_0\ar[d]_{i_0}&E_1\ar[d]_{i_1}\ar@<-0.7ex>[dl]_{s_0}\ar@<0.7ex>[dl]^{t_0}&E_2\ar[d]_{i_2}\ar@<-0.7ex>[dl]_{s_1}\ar@<0.7ex>[dl]^{t_1}&\ldots&E_{n-2}\ar[d]_{i_{n-2}}&E_{n-1}\ar[d]_{i_{n-1}}\ar@<-0.7ex>[dl]_{s_{n-2}}\ar@<0.7ex>[dl]^{t_{n-2}}&E_n\ar@<-0.7ex>[dl]_{s_{n-1}}\ar@<0.7ex>[dl]^{t_{n-1}}\\
    E_0^*&\ar@<-0.7ex>[l]_{\fpoly{s_0}}\ar@<0.7ex>[l]^{\fpoly{t_0}}\fpoly{E_1}&\ar@<-0.7ex>[l]_{\fpoly{s_1}}\ar@<0.7ex>[l]^{\fpoly{t_1}}\fpoly{E_2}&\ldots&E_{n-2}^*&E_{n-1}^*\ar@<-0.7ex>[l]_{\fpoly{s_{n-2}}}\ar@<0.7ex>[l]^{\fpoly{t_{n-2}}}\\
  }
\end{equation}
in~$\Set$ together with a structure of $(n-1)$-category on the $(n-1)$-graph
\[
\vxym{
  E_0^*&\ar@<-0.7ex>[l]_{\fpoly{s_0}}\ar@<0.7ex>[l]^{\fpoly{t_0}}\fpoly{E_1}&\ar@<-0.7ex>[l]_{\fpoly{s_1}}\ar@<0.7ex>[l]^{\fpoly{t_1}}\fpoly{E_2}&\ldots&E_{n-2}^*&E_{n-1}^*\ar@<-0.7ex>[l]_{\fpoly{s_{n-2}}}\ar@<0.7ex>[l]^{\fpoly{t_{n-2}}}
}
\]
Such a polygraph generates a free $n$-category, which will be noted~$E^*$, with
the previous $(n-1)$-category as underlying $(n-1)$-category and containing the
elements of~$E_n$ as $n$-cells, with source and target indicated by the
functions~$s_{n-1}$ and~$t_{n-1}$. We write~$E_n^*$ for its set of $n$-cells,
$i_n:E_n\to E_n^*$ for the canonical injection and $s_{n-1}^*,t_{n-1}^*:E_n^*\to
E_{n-1}^*$ for the source and target maps. An $(n+1)$-polygraph is then defined
as a diagram of the form
\[
\vxym{
  E_0\ar[d]_{i_0}&E_1\ar[d]_{i_1}\ar@<-0.7ex>[dl]_{s_0}\ar@<0.7ex>[dl]^{t_0}&E_2\ar[d]_{i_2}\ar@<-0.7ex>[dl]_{s_1}\ar@<0.7ex>[dl]^{t_1}&\ldots&E_{n-2}\ar[d]_{i_{n-2}}&E_{n-1}\ar[d]_{i_{n-1}}\ar@<-0.7ex>[dl]_{s_{n-2}}\ar@<0.7ex>[dl]^{t_{n-2}}&E_n\ar@<-0.7ex>[dl]_{s_{n-1}}\ar@<0.7ex>[dl]^{t_{n-1}}\ar[d]_{i_n}&E_{n+1}\ar@<-0.7ex>[dl]_{s_{n}}\ar@<0.7ex>[dl]^{t_{n}}\\
  E_0^*&\ar@<-0.7ex>[l]_{\fpoly{s_0}}\ar@<0.7ex>[l]^{\fpoly{t_0}}\fpoly{E_1}&\ar@<-0.7ex>[l]_{\fpoly{s_1}}\ar@<0.7ex>[l]^{\fpoly{t_1}}\fpoly{E_2}&\ldots&E_{n-2}^*&E_{n-1}^*\ar@<-0.7ex>[l]_{\fpoly{s_{n-2}}}\ar@<0.7ex>[l]^{\fpoly{t_{n-2}}}&E_n^*\ar@<-0.7ex>[l]_{\fpoly{s_{n-1}}}\ar@<0.7ex>[l]^{\fpoly{t_{n-1}}}\\
}
\]
in~$\Set$, together with the structure of $n$-category on the $n$-graph at the
bottom of the diagram, such that
\[
s_{n-1}^*\circ s_n=s_{n-1}^*\circ t_n
\qqtand
t_{n-1}^*\circ s_n=t_{n-1}^*\circ t_n
\]
We refer the reader to~\cite{burroni:higher-word} for a detailed definition and
to~\cite{guiraud-malbos:higher-cat-fdt} for a more abstract equivalent
definition.

Given an $n$-polygraph~\eqref{eq:n-pol}, the elements of~$E_k$ are called
\emph{$k$-generators}. Notice that~$E_0^*$ is the free 0-category (\ie set)
on~$E_0$, so~$E_0^*$ is isomorphic to~$E_0$, and we might as well decide that
they are equal: we have implicitly done this assumption in our presentation of
$n$\nbd{}dimensional rewriting systems in order to ease the explanation. Given
two $n$\nbd{}polygraphs~$E$ and~$F$, a morphism between them consists of a
family \hbox{$(f_k:E_k\to F_k)_{0\leq k\leq n}$} of functions such that
$s_k\circ f_{k+1}=f_k^*\circ s_k$ for every index~$k$, and we write~$\nPol{n}$
for the category of $n$\nbd{}polygraphs. Every $m$-polygraph~$E$ admits an
underlying $n$\nbd{}polygraph (with~$m>n$), often written~$\polytrunc{E}{n}$,
thus inducing a forgetful functor~$\nPol{m}\to\nPol{n}$.

\section{Presentations of $n$-categories}
\label{sec:pres-n-cat}
A crucial application of polygraphs relies on the fact that they provide us with
a notion of \emph{presentation} of $n$-categories: an $(n+1)$\nbd{}polygraph~$E$
can be seen as a description of an $n$-category by generators (the elements of
$E_k$ with $0\leq k\leq n$ being generators for $k$\nbd{}cells) quotiented by
relations (the elements of~$E_{n+1}$). More precisely, an $(n+1)$-polygraph~$E$
generates an $(n+1)$\nbd{}category~$E^*$. The $n$-category \emph{presented} by
the polygraph~$E$, written~$\overline{E^*}$ is obtained by quotienting the
underlying $n$-category of~$E^*$ by the relation identifying two
$n$\nbd{}cells~$f$ and~$g$ whenever there exists an $(n+1)$\nbd{}cell
$\alpha:f\to g$ in~$E^*$. More generally, one says that an $n$-category~$\C$ is
\emph{presented} by an $(n+1)$\nbd{}polygraph~$E$ when~$\C$ is isomorphic
to~$\overline{E^*}$.

Presentations of categories are useful in the sense that they can provide us
with small (even sometimes finite) descriptions of categories. This idea of
presentation generalizes the well-known notion of presentation of a monoid (or a
group). We begin by recalling this simple and well studied setting.

\subsection{Presentations of categories}
\subsubsection{Presentations of monoids}
\label{sec:mon-pres}
One of the main applications of string rewriting systems, that we are going to
generalize here, is the construction of presentations of a monoid.

\begin{defi}
  A \emph{presentation}~$\pres\Sigma R$ of a monoid~$M$ consists of
  \begin{itemize}
  \item a set~$\Sigma$ of \emph{generators},
  \item a set~$R\subseteq \Sigma^*\times\Sigma^*$ of \emph{relations}
  \end{itemize}
  where~$\Sigma^*$ denotes the free monoid over~$\Sigma$, such that the
  monoid~$M$ is isomorphic to~$\Sigma^*/\equiv_R$, where~$\equiv_R$ is the
  smallest congruence (\wrt concatenation of words) on~$\Sigma^*$
  containing~$R$.
\end{defi}

\begin{exa}
  \label{ex:pres-mon}
  The additive monoid~$\N$ admits the presentation with one generator~$a$ and no
  relation~$\pres a{}$. Similarly,
  \[
  \N/2\N
  \cong
  \pres a{aa=1}
  \qquad\qquad
  \N\times\N
  \cong
  \pres{a,b}{ba=ab}
  \]
  and, if we write~$\mathfrak{S}_n$ for the monoid of symmetries on $n$
  elements,
  \[
  \mathfrak{S}_n\cong\pres{\sigma_1,\ldots,\sigma_n}{\sigma_i\sigma_{i+1}\sigma_i=\sigma_{i+1}\sigma_i\sigma_{i+1},\
    \sigma_i^2=1,\ \sigma_i\sigma_j=\sigma_j\sigma_i}
  \]
\end{exa}


How does one construct a presentation~$\pres\Sigma R$ of a given monoid~$M$? The
main difficulty consists in showing the isomorphism
$M\cong(\Sigma^*/\equiv_R)$ which is a priori difficult because the second
monoid is described as terms modulo a congruence. One of the great usefulness of
string rewriting systems is that, when they are convergent (\ie terminating and
confluent), they provide one with a notion of normal form of terms modulo the
congruence generated by the rewriting rules. In order to show the isomorphism
$M\cong(\Sigma^*/\equiv_R)$, one can thus try to follow the following recipe:
\begin{enumerate}
\item Orient the relations in~$R$ in order to get a string rewriting system on
  the alphabet~$\Sigma$.
\item Show that the rewriting system is terminating.
\item Compute the critical pairs of the rewriting system and show that they are
  joinable.
\item The two previous points ensure that the rewriting system is
  convergent. Compute the normal forms and show that they are in bijection with
  the elements of~$M$ in a way compatible with multiplication and unit.
\end{enumerate}
It is well known that the joinability of critical pairs checked in (3) implies
the local confluence of the rewriting system. Its confluence can then be deduced
by Newman's lemma provided that it is terminating (2). The compatibility
condition of (4) can be made explicit as follows. We write~$(M,\times,1)$ for
the operations of the monoid~$M$, we write~$\overline{w}$ for the normal form of
a word~$w\in\Sigma^*$, and
set~$\overline{\Sigma^*}=\setof{\overline{w}}{w\in\Sigma^*}$. One should provide
a pair of functions~$f:\overline{\Sigma^*}\to M$
and~$g:M\to\overline{\Sigma^*}$, which are mutually inverse morphisms of
monoids:
\[
f(\overline{\overline u\otimes\overline v})=f(\overline u)\times f(\overline v)
\qquad
f(\overline\varepsilon)=1
\qquad
g(m\times n)=g(m)\otimes g(n)
\qquad
g(1)=\overline\varepsilon
\]
for every~$u,v\in\Sigma^*$ and $m,n\in M$, where~$\otimes$ denotes concatenation
and~$\varepsilon$ the empty word.

\begin{exa}
  \label{ex:mon-pres}
  We can show that the additive monoid~$M=\N\times(\N/2\N)$ admits the
  presentation~$\pres{a,b}{ba=ab, 1=bb}$, where $1$ denotes the empty word:
  \begin{enumerate}
  \item We orient the rules as follows:
    \[
    ba\To ab
    \qqtand
    bb\To 1
    \]
  \item The system is terminating: introduce a suitable measure on words based
    on the fact that the rules respectively decrease the number of $a$ after a
    $b$, and the total number of $b$ in a word.
  \item The two critical pairs are joinable:
    \[
    \vxym{
      &\ar@{=>}[dl]bba\ar@{=>}[dr]&\\
      a\ar@{=}[dr]&&bab\ar@{=>}[d]\\
      &a&\ar@{=>}[l]abb\\
    }
    \qquad\qquad\qquad\qquad
    \vxym{
      &\ar@{=>}[dl]bbb\ar@{=>}[dr]&\\
      b\ar@{=}[dr]&&\ar@{=}[dl]b\\
      &b&\\
    }
    \]
  \item The normal forms are words of the form $a^nb^m$ with~$n\in\N$
    and~$m\in\{0,1\}$. These are obviously in bijection with elements of
    $\N\times(\N/2\N)$ and the bijection can be shown to be a morphism of
    monoids (see below).
  \end{enumerate}
\end{exa}

Since normal forms are canonical representatives of elements of the free
monoid~$\Sigma^*$ over the alphabet~$\Sigma$ modulo the congruence~$\equiv_R$,
in order to define the morphism~$f$ witnessing the bijection (4), by universal
property of the free monoid it is enough to define~$f$ on letters in~$\Sigma$,
extend it as a morphism~$f:\Sigma^*\to M$ and show that the images of two words
equivalent \wrt $\equiv_R$ are equal, which can be checked by showing that the
image by~$f$ of any left member of any rule is equal the corresponding right
member. Of course, the explicit construction of~$g$ is not necessary and it is
equivalent (and sometimes more natural) to show that~$f$ is both surjective and
injective.

\begin{exa}
  The last step of Example~\ref{ex:mon-pres} can thus be shown as follows:
  \begin{itemize}
  \item[(4)] We define a morphism~$f:\overline{\Sigma^*}\to M$ as follows. The
    morphism~$f$ is defined on letters by~$f(a)=(1,0)$, $f(b)=(0,1)$. The
    induced morphism~$f:\Sigma^*\to M$ is compatible with the
    congruence~$\equiv_R$:
    \[
    f(ba)=f(b)+f(a)=(0,1)+(1,0)=(1,1)=(1,0)+(0,1)=f(a)+f(b)=f(ab)
    \]
    and
    \[
    f(bb)=f(b)+f(b)=(0,1)+(0,1)=(0,0)=f(1)
    \]
    Conversely, $g$ is defined by $g((m,n))=a^mb^n$, which is a morphism of monoids,
    \ie
    \[
    g((m,n)+(m',n'))=a^{m+m'}b^{n+n'}=\overline{a^mb^na^{m'}b^{n'}}
    \]
    and the two morphisms are mutually inverse
    \[
    g\circ f(a^mb^n)=g((m,n))=a^mb^n
    \qtand
    f\circ g((m,n))=f(a^mb^n)=(m,n)
    \]
  \end{itemize}
\end{exa}

This general methodology can be straightforwardly adapted to $n$-polygraphs in
order to build presentations of $n$-categories. We detail a few interesting
particular cases in the following.

\begin{rem}
  Notice that the method described above is only a ``recipe'' for constructing
  presentations, which works (or can be adapted) in many cases. In fact, it has
  been shown by Squier~\cite{squier:word-problems,
    lafont-proute:homology-monoids}, using homological methods, that there
  exists monoids which admit finite presentations, but no finite presentation
  which can be oriented into a convergent rewriting system. These theoretical
  considerations aside, convergent presentations are sometimes difficult to
  build: in these cases, the following weaker method (used for example
  in~\cite{lafont:boolean-circuits}) can be tried:
  \begin{enumerate}
    \setcounter{enumi}{1}
  \item Introduce a notion of canonical form for words (which will play the role
    of previous normal forms).
  \item Show that every word reduces to a canonical form.
  \item Define a morphism~$f:\Sigma^*\to M$ which is compatible with rewriting
    rules as previously, and show that it induces a bijection between the
    canonical forms and the elements of~$M$.
  \end{enumerate}
  This variant has the advantage of not requiring to show termination, but it
  can lead one to have to consider many cases (many more than critical pairs when
  the system is convergent).
\end{rem}

\subsubsection{Presentations of categories}
2-polygraphs more generally present categories: the case of a presentation of a
monoid is the particular case where the presented category has only one
object. Namely, a 2-polygraph~$E$ describes a category which is the smallest
category with the 0\nbd{}generators in~$E_0$ as objects, containing the typed
1-generators in~$E_1$ as morphisms, quotiented by the relations in~$E_2$.

\begin{exa}
  \label{ex:pres-delta-cat}
  Consider the simplicial category~$\Delta$: its objects are natural
  integers~$n\in\N$, where~$n$ is seen as the totally ordered set with~$n$
  elements~$\{0,1,2,\ldots,n-1\}$ and morphisms~$f:m\to n$ are increasing
  functions. This category is presented by the 2\nbd{}polygraph~$E$ with
  $E_0=\N$, the set~$E_1$ contains the two families of 1-generators indexed by
  integers~\hbox{$n,i\in\N$}, such that~$0\leq i\leq n$,
  \[
  \mu_i^{n+1}:n+2\to n+1
  \qqtand
  \eta_i^n:n\to n+1
  \]
  and the set~$E_2$ of rewriting rules contains the families of 2-generators
  index by $n,i\in\N$
  \[
  \begin{array}{rcl}
    \mu^{n+1}_j\mu^{n+2}_i&\To&\mu^{n+1}_i\mu^{n+2}_{j+1}\hspace{1.5ex}\text{for $i\leq j$,}\\
    \eta^{n+1}_i\eta^n_j&\To&\eta^{n+1}_{j+1}\eta^n_i\hspace{4.2ex}\text{for $i\leq j$,}\\
    \mu_j^{n+1}\eta_i^{n+1}&\To&
    \begin{cases}
      \eta_i^n\mu_{j-1}^n&\text{for $i<j$,}\\
      \id_{n+1}&\text{for $i=j$ or $i=j+1$,}\\
      \eta^n_{i-1}\mu^n_j&\text{for $i>j+1$}\\
    \end{cases}
  \end{array}
  \]
  In order to show this, the same recipe as previously can be used: the
  rewriting system can be shown to be terminating and confluent, the normal
  forms being terms of the form
  \[
  \eta_{i_k}^{m-h+k-1}\circ\ldots\circ\eta_{i_1}^{m-h}\circ\mu_{j_h}^{m-h}\circ\ldots\circ\mu_{j_1}^{m-1}
  \qcolon
  m\to n
  \]
  such that $n=m-h+k$, $n>i_k>\ldots>i_1\geq 0$ and $0\leq j_h<\ldots<j_1<m$. We
  can then construct a functor~$F:\overline{E^*}\to\Delta$, which is defined as
  the identity on objects and is defined on 1-generators by
  \[
  F(\mu_i^n)=k\mapsto
  \begin{cases}
    k&\text{if $k\leq i$}\\
    k-1&\text{if $k>i$}\\
  \end{cases}
  \qtand
  F(\eta_i^n)=k\mapsto
  \begin{cases}
    k&\text{if $k<i$}\\
    k+1&\text{if $k\geq i$}\\
  \end{cases}
  \]
  The image of the left member of a rule by~$F$ can be checked to be equal to
  the image of the corresponding right member and it can be shown to be a
  bijection, see~\cite{maclane:cwm} (section VII.5) for details. A refined way
  to construct this presentation is hinted in Example~\ref{ex:simpl-2-pres}.
\end{exa}

\subsection{Presentations of 2-categories}
\subsubsection{Monoidal categories}
Before explaining how term rewriting systems provide us with a notion of
presentation of a Lawvere theory, we need to introduce the notion of monoidal
category. Formally, a monoidal category could be defined as a 2-category
(Definition~\ref{def:2-cat}) with exactly one 0-cell. However, in the same way
that a category with only one object can be reformulated as a monoid (\ie a set
with operations), a monoidal category is generally defined as a category with
operations.

\begin{defi}
  \label{def:mon-cat}
  A \emph{strict monoidal category}~$(\C,\otimes,I)$ consists of a category~$\C$
  together with
  \begin{itemize}
  \item a functor~$\otimes:\C\times\C\to\C$, called \emph{tensor product},
  \item an object~$I$ of~$\C$, called \emph{unit},
  \end{itemize}
  such that
  \begin{itemize}
  \item tensor product is associative: for every objects~$A$,$B$ and $C$,
    $(A\otimes B)\otimes C=A\otimes(B\otimes C)$
  \item units are neutral elements: for every object~$A$, $I\otimes A=A=A\otimes
    I$.
  \end{itemize}
\end{defi}

The tensor product~$f\otimes g$ of two morphisms~$f$ and~$g$ should be thought
as the morphism~$f$ ``in parallel'' with the morphism~$g$. This definition is a
particular case of a more general definition of monoidal
categories~\cite{maclane:cwm} (which is why it is called \emph{strict}), however
we will only consider this variant here. A monoidal category $(\C,\otimes,I)$ as
above gives rise to a 2-category, with only one 0-cell, the 1\nbd{}cells being
the objects of~$\C$, the 2-cells being the morphisms of~$\C$, vertical
composition being given by the composition of~$\C$ and horizontal composition
being given by~$\otimes$; and this induces an isomorphism between the
category of monoidal categories and the categories of 2-categories with one
object (this is a generalization of the situation between monoids and categories
described in Section~\ref{sec:2-rs}). In other words, a monoidal category is
``the same as'' a 2-category with one object and we can in particular reuse the
string-diagrammatic notation for morphisms. We thus sometimes implicitly
consider a monoidal category as a 2-category in the following.

Any cartesian category~$\C$ can be equipped with a structure of monoidal
category~$(\C,\times,1)$, tensor product being cartesian product and unit being
the terminal object -- for simplicity, we consider that cartesian products are
strictly associative, which is always true up to equivalence of categories, but
this assumption could have been dropped if we had worked with the more general
notion of (non-strict) monoidal category. The usual notion of monoid can be
generalized from the monoidal category $(\Set,\times,1)$ to any monoidal
category as follows.

\begin{defi}
  \label{def:monoid}
  A \emph{monoid} $(M,\mu,\eta)$ in a monoidal category~$(\C,\otimes,I)$
  consists of an object~$M$ equipped with two morphisms
  \[
  \mu:M\otimes M\to M
  \qtand
  \eta:I\to M
  \]
  such that the diagrams
  \[
  \vxym{
    M\otimes M\otimes M\ar[d]_{\id_M\otimes M}\ar[r]^-{\mu\otimes\id_M}&M\otimes M\ar[d]^\mu\\
    M\otimes M\ar[r]_\mu&M
  }
  \qtand
  \vxym{
    I\otimes M\ar[dr]_{\id_M}\ar[r]^{\eta\otimes\id_M}&M\otimes M\ar[d]_\mu&\ar[l]_{\id_M\otimes\eta}\ar[dl]^{\id_M}M\otimes I\\
    &M&\\
  }
  \]
  commute. String-diagrammatically,
  \[
  \strid{mon_assoc_l}
  =
  \strid{mon_assoc_r}
  \qtand
  \strid{mon_neutral_l}
  =
  \strid{mon_neutral_c}
  =
  \strid{mon_neutral_r}
  \]
\end{defi}

\begin{exa}
  A monoid in~$(\Set,\times,1)$ is a monoid in the usual sense, a monoid
  in~$(\Cat,\times,1)$ is a monoidal category.
\end{exa}

As mentioned before, the tensor product intuitively expresses the possibility of
having two morphisms done in parallel. One sometimes also needs to be able to
``swap'' two morphisms, which is formalized by the notion of symmetric monoidal
category.

\begin{defi}
  A \emph{symmetric monoidal category} $(\C,\otimes,I,\gamma)$ consists of a
  monoidal category~$(\C,\otimes,I)$ together with a family~$\gamma$ of
  isomorphisms~$\gamma_{A,B}:A\otimes B\to B\otimes A$, indexed by pairs of
  objects~$A,B$ of~$\C$ and called \emph{symmetry}, which is natural in the
  sense that for every morphisms~$f:A\to A'$ and~$g:B\to B'$,
  \[
  \vxym{
    A\otimes B\ar[d]_{\gamma_{A,B}}\ar[r]^{f\otimes g}&A'\otimes B'\ar[d]^{\gamma_{A',B'}}\\
    B\otimes A\ar[r]_{g\otimes f}&B'\otimes A'
  }
  \]
  and such that for every objects~$A$, $B$ and $C$ the diagrams
  \[
  \xymatrix@C-6 pt{%
    A\otimes B\otimes C\ar[rr]^{\gamma_{A,B\otimes C}}\ar[dr]_{\gamma_{A,B}\otimes\id_C}&&B\otimes C\otimes A\\
    &B\otimes A\otimes C\ar[ur]_{\id_B\otimes\gamma_{A,C}}&
  }
  \quad
  \xymatrix@C-6 pt{%
    A\otimes B\otimes C\ar[dr]_{\id_A\otimes\gamma_{B,C}}\ar[rr]^{\gamma_{A\otimes B,C}}&&C\otimes A\otimes B\\
    &A\otimes C\otimes B\ar[ur]_{\gamma_{A,C}\otimes \id_B}&
  }
  \]
  commute, and for every objects~$A$ and~$B$,
  \[
  \gamma_{B,A}\circ\gamma_{A,B}\qeq\id_{A\otimes B}
  \]
\end{defi}

Again, every cartesian category can be equipped with a symmetry using the
canonical isomorphisms~$A\otimes B\cong B\otimes A$ induced by the cartesian product. The
notion of commutative monoid can thus be generalized to any symmetric monoidal
category as follows.

\begin{defi}
  A \emph{commutative monoid} $(M,\mu,\eta)$ in a symmetric monoidal
  category $(\C,\otimes,I,\gamma)$ is a monoid such that
  \[
  \mu\circ\gamma_{M,M}\qeq\mu
  \]
  String-diagrammatically,
  \[
  \strid{cmon_l}
  \qeq
  \strid{cmon_r}
  \]
\end{defi}

The notion of \emph{commutative comonoid}~$(M,\delta,\varepsilon)$
with~$\delta:M\to M\otimes M$ and~$\varepsilon:M\to I$ is defined
dually. Interestingly, this notion allows us to characterize cartesian
categories among monoidal categories:

\begin{prop}
  \label{prop:cart-comon}
  A monoidal category~$(\C,\otimes,I)$ is cartesian, with~$\otimes$ as cartesian
  product and~$I$ as terminal object, if and only if the category can be
  equipped with a symmetry~$\gamma$ and every object~$A$ can be equipped with a
  structure of commutative comonoid~$(A,\delta_A,\varepsilon_A)$ which is
  natural in the sense that for every morphism~$f:A\to B$ the diagrams
  \[
  \vxym{
    A\ar[d]_{\delta_A}\ar[r]^f&B\ar[d]^{\delta_B}\\
    A\otimes A\ar[r]_{f\otimes f}&B\otimes B
  }
  \qtand
  \vxym{
    A\ar[dr]_{\varepsilon_A}\ar[rr]^f&&B\ar[dl]^{\varepsilon_B}\\
    &I
  }
  \]
  commute; string-diagrammatically,
  \[
  \strid{nat_delta_l}
  =
  \strid{nat_delta_r}
  \qquad\qquad
  \strid{nat_eps_l}
  =
  \strid{nat_eps_r}
  \]
\end{prop}

This result was considered for a long time as folklore in category theory. A
detailed proof of this fact is not really difficult and can for example be found
in~\cite{mellies:cat-sem-ll}. The intuition behind this result is however very
enlightening: a cartesian category is a monoidal category in which any
object~$A$ can be duplicated (by~$\delta_A:A\to A\otimes A$), erased
(by~$\varepsilon_A:A\to I$) or swapped with another object~$B$
(by~$\gamma_{A,B}:A\otimes B\to B\otimes A$)! These operations can also be
thought as the analogous of contraction, weakening and exchange rules
respectively in sequent calculus. We use this result below in order to show that
in fact presentations of Lawvere theories are a particular case of presentations
of 2-categories.

\subsubsection{Presentations of Lawvere theories}
\label{sec:pres-law}
We have seen that string rewriting systems correspond to presentations of
monoids. Similarly, we may wonder: what are term rewriting systems presentations of? It turns
out that the right answer for this are particular categories, called Lawvere
theories, which were introduced by Lawvere in his thesis~\cite{lawvere:phd}.

\begin{defi}
  A \emph{Lawvere theory} is a cartesian category whose objects are
  integers~$\N$ such that~$0$ is a terminal object and product is given on
  objects by addition.
\end{defi}

\begin{exa}
  \label{ex:Nmat}
  A Lawvere theory~$\Nmat$ of $\N$-valued matrices can be defined as
  follows. Its objects are natural numbers and a morphism~$M:m\to n$ is a
  matrix~$M$ of size $n\times m$ (with $n$ rows and $m$ columns) with coefficients in~$\N$. The composition
  $N\circ M:m\to p$ of two morphisms~$M:m\to n$ and \hbox{$N:n\to p$} is given
  by the usual composition~$N\circ M=NM$ of matrices and identities are usual
  identity matrices. This category is cartesian with cartesian product defined
  on objects by addition and the cartesian product of two morphisms~$M:m\to m'$
  and \hbox{$N:n\to n'$} is their direct sum, \ie the matrix~$M\times N:m+n\to
  m'+n'$ defined as
  \[
  \left(
    \begin{matrix}
      M&0\\
      0&N
    \end{matrix}
  \right)
  \]
  using a block representation of the matrix.
\end{exa}

The link with term rewriting systems can be explained as follows. Suppose given
a signature~$(\Sigma,a)$. This signature induces a Lawvere theory~$\Law\Sigma$
whose morphisms~$t:m\to n$ are $n$-uples $(t_1,\ldots,t_n)$ of terms using
variables in~$X_m=\{x_1,\ldots,x_m\}$. The composition \hbox{$u\circ t:m\to p$}
of two morphisms~$t:m\to n$ and~$u:n\to p$ is the morphism~$v$ defined by
substitution as
\[
(u_1[t_1/x_1,\ldots,t_n/x_n]\ ,\ \ldots\ ,\ u_p[t_1/x_1,\ldots,t_n/x_n])
\]
and the identity on an object $n$ is $(x_1,\ldots,x_n)$. This category is
cartesian, in particular the structure of commutative comonoid $(1,\delta:1\to
2,\varepsilon:1\to 0)$ on the object~$1$ given by
Proposition~\ref{prop:cart-comon} is defined by~$\delta=(x_1,x_1)$,
$\varepsilon=()$ and the symmetry on~$1$ is~$\gamma_{1,1}=(x_2,x_1)$.

Suppose given a term rewriting system~$R$ on~$\Sigma$. We write~$\equiv_R$ for
the smallest congruence on terms containing the rules in~$R$. A term rewriting
system~$R$ \emph{presents} a Lawvere theory~$\C$ when~$\C$ is isomorphic to
$\Law\Sigma/\equiv_R$ (the Lawvere theory generated by the signature~$\Sigma$
whose morphisms are quotiented by the equivalence relation~$\equiv_R$).  Notice
that, in order to show the isomorphism, one can try to apply the method given in
Section~\ref{sec:mon-pres}.

\begin{exa}
  \label{ex:mon-trs}
  Consider a term rewriting system corresponding to commutative monoids: the
  signature contains two symbols~$m$ of arity $2$ and~$e$ of arity~$0$, and
  there are four rewriting rules
  \[
  m(m(x,y),z)\To m(x,m(y,z))
  \qquad
  m(e,x)\To x
  \qquad
  m(x,e)\To x
  \qquad
  m(x,y)\To m(y,x)
  \]
  This rewriting system presents the Lawvere theory~$\Nmat$ of $\N$-valued
  matrices described in Example~\ref{ex:Nmat}. In order to show this, one could
  try to show that the rewriting system is convergent, but this fails
  immediately: the rewriting system is not terminating because of the last rule
  for commutativity. It is however possible to overcome this difficulty by
  working on terms modulo commutativity, but we do not detail this here since a
  more satisfactory way of proving the result is mentioned in
  Example~\ref{ex:mon-poly}.
  In order to provide an intuition about the result we shall however describe
  the functor~$F:(\Law\Sigma/\!\equiv_R)\to\Nmat$ witnessing the isomorphism
  between the two categories, which can be shown to preserve cartesian
  products. The functor~$F$ is the identity on objects. Suppose given a morphism
  $t:m\to n$, constituted of an $n$\nbd{}uple of terms $(t_1,\ldots,t_n)$ with
  variables in~$X_m$. The image~$F(t)$ of the morphism~$t$ is the matrix~$M$ of
  size $n\times m$ such that the entry~$M_{i,j}$ is the number of occurrences of
  the variable~$x_j$ in the term~$t_i$. For example, the image of the morphism
  \[
  m(m(x_1,x_1),x_2),\ e,\ x_2
  \qcolon
  2 \to 3
  \]
  is the matrix
  \[
  \left(
    \begin{matrix}
      2&1\\
      0&0\\
      0&1
    \end{matrix}
  \right)
  \]
  The functor~$F$ can be shown to be compatible with the rewriting rules, \ie
  the images of two morphisms equivalent modulo~$\equiv_R$ are equal, for
  example the morphism
  \[
  m(x_1,m(x_2,x_1)),\ m(e,e),\ m(e,x_2)
  \qcolon
  2\to 3
  \]
  is equivalent to the morphism above and has the same image under~$F$. Finally,
  the functor~$F$ can be shown to be an isomorphism.
\end{exa}

The result above is quite interesting because it provides one with a concrete
description of the category generated by the rewriting system which describes
the notion of commutative monoid in a cartesian category: from the result in the
example above, it is easy to deduce the somewhat surprising fact that the
category~$\Nmat$ ``impersonates'' the notion of monoid in the sense that a
monoid in a cartesian category~$\C$ is ``the same as'' a functor $\Nmat\to\C$
which preserves cartesian products (more formally, the category of
product-preserving functors from~$\Nmat$ to~$\C$ and cartesian natural
transformations between them is equivalent to the category of monoids in~$\C$
and morphisms of monoids).

We have explained in Section~\ref{sec:trs} that a 3-polygraph~\eqref{eq:rs-3}
with~$E_0=\{*\}$ and \hbox{$E_1=\{1\}$} reduced to exactly one element and~$t_1$
indicating that all the 2\nbd{}generators have exactly one output reformulate
linear term rewriting systems. At the light of Proposition~\ref{prop:cart-comon},
if we allow 2-generators with multiple outputs, (non-linear) term rewriting
systems can be represented by adding 2-generators
\[
\delta:1\To 1\otimes 1
\qquad\qquad
\varepsilon:1\To 0
\qquad\qquad
\gamma:1\otimes 1\To 1\otimes 1
\]
which will allow explicit duplication, erasure and swapping of variables. This
was first formalized in~\cite{burroni:higher-word}:

\begin{prop}
  Suppose given a term rewriting system~$R$ on a signature~$(\Sigma,a)$. The
  Lawvere theory $\Law\Sigma/\equiv_R$ generated by the rewriting system (seen
  as a monoidal category and thus as a 2-category) is isomorphic to the
  2-category~$\overline{E^*}$ presented by the 3-polygraph~$E$ with
  \[
  E_0=\{*\}
  \qquad\qquad
  E_1=\{1\}
  \qquad\qquad
  E_2=\Sigma\uplus\{\delta,\varepsilon,\gamma\}
  \qquad\qquad
  E_3=R\uplus C
  \]
  with the expected source and target maps, where the rules in~$C$ express
  the fact that $(1,\delta,\varepsilon)$ has a natural structure of commutative
  comonoid with~$\gamma$ generating the symmetry.
\end{prop}

We shall make the source and target maps and the rules in~$C$ a bit more
explicit. If we write~$n$, as previously, for the tensor of $n$ copies of the
object~$1$, the map~$s_1:E_2\to E_1^*$ is defined by
\[
s_1(f)=a(f)
\qquad
s_1(\delta)=1
\qquad
s_1(\varepsilon)=1
\qquad
s_1(\gamma)=2
\]
for any symbol~$f\in\Sigma$ of arity $a(f)$, and similarly~$t_1$ is defined by
\[
t_1(f)=1\qquad
\qquad
t_1(\delta)=2
\qquad
t_1(\varepsilon)=0
\qquad
t_1(\gamma)=2
\]
The rules in~$C$ express that $(1,\delta,\varepsilon)$ should be a commutative
comonoid:
\[
(\delta\otimes\id_1)\circ\delta\TO (\id_1\otimes\delta)\circ\delta
\qquad
(\varepsilon\otimes\id_1)\circ\delta\TO\id_1
\qquad
(\id_1\otimes\varepsilon)\circ\delta\TO\id_1
\qquad
\gamma\circ\delta\TO\delta
\]
and the morphisms corresponding to generators should be natural \wrt every other
morphism corresponding to a generator, \eg for every element~$f:1\To 1$
of~$E_2$, $C$~should also contain the rules
\[
\delta\circ f\TO(f\otimes f)\circ\delta
\qquad
\varepsilon\circ f\TO\varepsilon
\qquad
\gamma\circ(f\otimes\id_1)\TO(\id_1\otimes f)\circ\gamma
\qquad
\gamma\circ(\id_1\otimes f)\TO(f\otimes\id_1)\circ\gamma
\]
(similar rules should also be introduced for 2-generators~$f$ whose arity or
coarity is different from~$1$,
see~\cite{burroni:higher-word}). Finally, the rules in~$R$ have source and
target the terms corresponding to their left and right member ``translated'' to
explicit manipulations of variables (by using the morphisms~$\delta$, $\varepsilon$
and~$\gamma$) followed by a linear term. For example, if~$\Sigma$ contains two
symbols~$f$ and~$g$ of arity~$2$, a rule $f(g(x_2,x_1),x_2)\TO g(x_1,x_1)$ will
be translated as
\[
f\circ(g\otimes\id_1)\circ(\gamma\otimes\id_1)\circ(\id_1\otimes\delta)\qTO
g\circ(\delta\otimes\varepsilon)
\]
which is maybe best understood with the corresponding string diagrammatic
notation
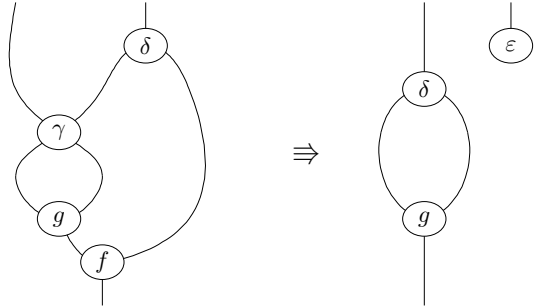
\[
\strid{ex_ltrs_l}
\qTO
\strid{ex_ltrs_r}
\]
(which is again very close to the usual notation of sharing graphs). This way of
representing variables in terms has been the starting point for numerous series
of works trying to formalize and axiomatize variable binding, \eg
\cite{fiore-plotkin-turi:as-vb}.

\begin{exa}
  \label{ex:mon-poly}
  The 3-polygraph corresponding to the term rewriting system of monoids
  introduced in Example~\ref{ex:mon-trs} is the theory of bicommutative
  bialgebras, and this polygraph can be shown directly to present the
  2-category~$\Nmat$ using rewriting techniques~\cite{lafont:boolean-circuits,
    mimram:phd, mimram:first-order-causality}. A particularly interesting point
  is that since swapping of variables is now explicit, the rule expressing
  commutativity of the monoid is now expressed as
  \[
  \strid{comm_l}
  \qTO
  \strid{comm_r}
  \]
  where~$\gamma$ and~$\mu$ are the 2-generators corresponding to symmetry and
  multiplication respectively: the rule is not anymore a priori an obstacle to
  termination (contrarily to the approach of Example~\ref{ex:mon-trs}) because
  it makes the number of~$\gamma$ generators decrease in the morphisms. However,
  no convergent rewriting system for bicommutative bialgebras is currently
  known.
\end{exa}

\subsubsection{Presentations of 2-categories}
3-polygraphs provide us with a general notion of presentation by 0-, 1- and
2-generators and relations (the 3-generators) of a 2-category. Many examples of
such presentations have been studied by
Lafont~\cite{lafont:boolean-circuits}. We briefly recall here some fundamental
examples that he discovered. The notion of critical pair of the corresponding
rewriting systems will be formally introduced in Section~\ref{sec:cp-3}.

\begin{exa}
  \label{ex:simpl-2-pres}
  In Example~\ref{ex:pres-delta-cat}, we have recalled the presentation of the
  simplicial category~$\Delta$. An unsatisfactory point about this presentation
  is that it is infinite (it has an infinite number of generators and of
  relations). Interestingly, one can provide a finite presentation of this
  category by considering it as a 2-category (more precisely as a monoidal
  category). Namely, a tensor product~$\otimes$ on the category can be defined
  on objects~$m$ and~$n$ by~$m\otimes n=m+n$ and on morphisms~$f:m\to m'$
  and~$g:n\to n'$ as the morphism $f\otimes g:m+n\to m'+n'$
  \[
  f\otimes g
  \qeq
  i\mapsto
  \begin{cases}
    f(i)&\text{if $i<m$}\\
    g(i-m)+m'&\text{if $i\geq m$}
  \end{cases}
  \]
  This induces a structure of monoidal category on~$\Delta$, with~$0$ as
  unit. This monoidal category admits a presentation by the 3-polygraph~$E$
  corresponding to the theory of monoids (see Definition~\ref{def:monoid}),
  with~$E_0=\{*\}$, $E_1=\{1\}$,
  \hbox{$E_2=\{\mu:2\to 1, \eta:0\to 1\}$} and the 3\nbd{}generators being
  \[
  a:
  \mu\circ(\mu\otimes\id_1)
  \TO
  \mu\circ(\id_1\otimes\mu)
  \qquad\qquad
  l:
  \mu\circ(\eta\otimes\id_1)
  \TO
  \id_1
  \qquad\qquad
  r:
  \mu\circ(\id_1\otimes\eta)
  \TO
  \id_1
  \]
  with the following string diagrammatic representation
  \[
  \strid{rs_mon_assoc_l}
  \overset{a}\TO
  \strid{rs_mon_assoc_r}
  \qquad\qquad
  \strid{rs_mon_unit_l}
  \overset{l}\TO
  \strid{rs_mon_unit_c}
  \qquad\qquad
  \strid{rs_mon_unit_r}
  \overset{r}\TO
  \strid{rs_mon_unit_c}
  \]
  This rewriting system is terminating and its five critical pairs
  \[
  \strid{rs_mon_cp1}
  \qquad
  \strid{rs_mon_cp2}
  \qquad
  \strid{rs_mon_cp3}
  \qquad
  \strid{rs_mon_cp4}
  \qquad
  \strid{rs_mon_cp5}
  \]
  are joinable (see Section~\ref{sec:contexts} for a precise definition of the
  critical pairs in a 3-polygraph). The normal forms are described by the
  following grammar
  \[
  \phi
  \qgramdef
  \strid{rs_mon_nf_eta}
  \qtor
  \strid{rs_mon_nf_id}
  \qtor
  \strid{rs_mon_nf_mu}
  \]
  One can define a functor~$F:\overline{E^*}\to\Delta$ as follows. The
  functor is the identity on objects. The image of~$\mu:2\to 1$ is the constant
  function~$m:2\to 1$ (which to~$0$ and~$1$ associates~$1$) and the image
  of~$\eta:0\to 1$ is the constant function~$e:0\to 1$. Finally, the functor~$F$
  can be shown to be an isomorphism (it is easy to show that there is a
  one-to-one correspondence between normal forms in~$E$ and functions
  in~$\Delta$). The fact that~$F$ is full means that every function can be
  expressed as a composite and tensor of~$m$ and~$e$. For example, the
  function~$f:4\to 3$ whose graph is pictured on the left (such that~$f(0)=1$,
  $f(1)=1$, $f(2)=1$ and~$f(3)=2$) is the image of the morphism on the right:
  \[
  \strid{ex_fun_i}
  \qqtext{is the image of}
  \strid{ex_fun}
  \]
  Of course there are many ways to express a given function~$f$ using~$m$
  and~$e$. The fact that the functor is faithful expresses the fact that if two
  morphisms obtained by composing and tensoring~$\mu$ and~$\eta$ have the same
  image by~$F$, then they are equivalent modulo the rewriting rules. Notice that
  the result of Example~\ref{ex:pres-delta-cat} can be recovered by defining
  $\mu_i^{n+1}=\id_i\otimes\mu\otimes\id_{n-i}$
  and~$\eta_i^n=\id_i\otimes\eta\otimes\id_{n-i}$.
\end{exa}

\begin{exa}
  \label{ex:pres-bij}
  The monoidal category~$\Bij$ is defined similarly as the simplicial
  category~$\Delta$ excepting that its morphisms are bijective functions. This
  monoidal category admits a presentation by the 3-polygraph~$E$ corresponding
  to the theory of symmetries, with \hbox{$E_0=\{*\}$}, $E_1=\{1\}$,
  \hbox{$E_2=\{\gamma:2\to 2\}$} and the 3\nbd{}generators being
  \[
  (\gamma\otimes\id_1)\circ(\id_1\otimes\gamma)\circ(\gamma\otimes\id_1)
  \TO
  (\id_1\otimes\gamma)\circ(\gamma\otimes\id_1)\circ(\id_1\otimes\gamma)
  \qquad\qquad
  \gamma\circ\gamma\TO\id_2
  \]
  whose string diagrammatic representation is
  \[
  \strid{yb_l}
  \TO
  \strid{yb_r}
  \qquad\qquad\qquad\qquad
  \strid{sym_l}
  \TO
  \strid{sym_r}
  \]
  The monoid of endomorphisms of an object~$n$ in the category~$\Bij$ is the
  symmetric group~$\mathfrak{S}_n$ (seen as a monoid). In this sense, the
  polygraph above provides a finite presentation of all the symmetric groups at
  once.
\end{exa}

Interestingly, the preceding rewriting system can be shown to be convergent
even though it has an infinite number of critical pairs. Namely, it has the
three following obvious critical pairs
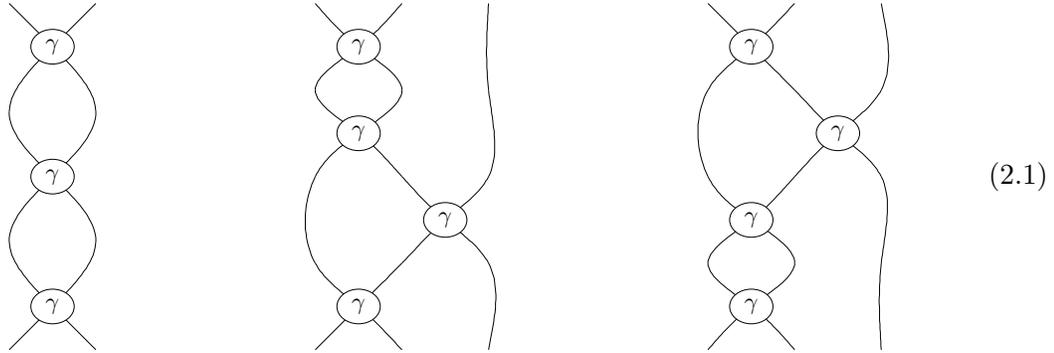
\begin{equation}
  \label{eq:ex-cp}
  \strid{gamma_cp_1}
  \qquad\qquad\qquad
  \strid{gamma_cp_2}
  \qquad\qquad\qquad
  \strid{gamma_cp_3}
\end{equation}
Moreover, for every morphism $\phi:(1+m)\to(1+n)$, the morphism~\eqref{eq:cp}
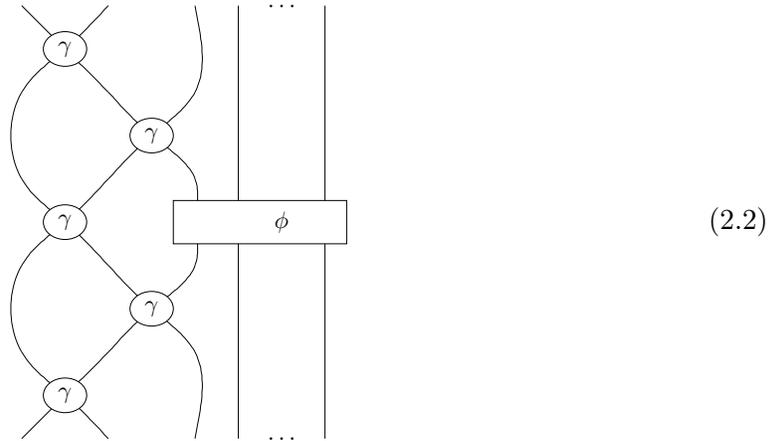
\begin{equation}
  \label{eq:cp}
  \strid{gamma_cp}
\end{equation}
can be rewritten in two different ways, giving rise to an infinite number of
critical pairs. Yet, the rewriting system can be shown to be
convergent~\cite{lafont:boolean-circuits}. This contrasts with string or term
rewriting systems, which always admit a finite number of critical pairs when
they are finite. The difference here seems to come essentially from the fact
that the generator $\gamma$ has multiple outputs as well as multiple inputs.

\bigskip

More recently, these tools have also been applied to more unexpected fields of
computer science. For example, the author has proposed a 3-polygraph presenting
a monoidal category of game semantics for first order propositional
logic~\cite{mimram:first-order-causality}.

\section{Representing 2-dimensional critical pairs}
\label{sec:repr-2-cp}
One of the main achievement of rewriting theory is to provide us with algorithms
to compute the critical pairs~\cite{baader-nipkow:trat}, which are at the basis
of many advanced tools, to automatically check the confluence of rewriting
systems or do Knuth-Bendix completions for example. The fact that the number of
critical pairs might be infinite for a finite rewriting system seems to indicate
that there is little hope to extend these nice techniques to higher
dimensions. We introduce here new theoretical tools in order to overcome this
difficulty.

In 3-dimensional rewriting systems, it turns out that we can nevertheless
recover a finite description of the critical pairs if we allow ourselves to
consider a more general notion of critical pair. For example, in the case of the
presentation of the category~$\Bij$ given in Example~\ref{ex:pres-bij}, the
rewriting system admits a finite number of ``critical pairs'' if we consider the
diagram on the left of~\eqref{eq:ccp} as a morphism:
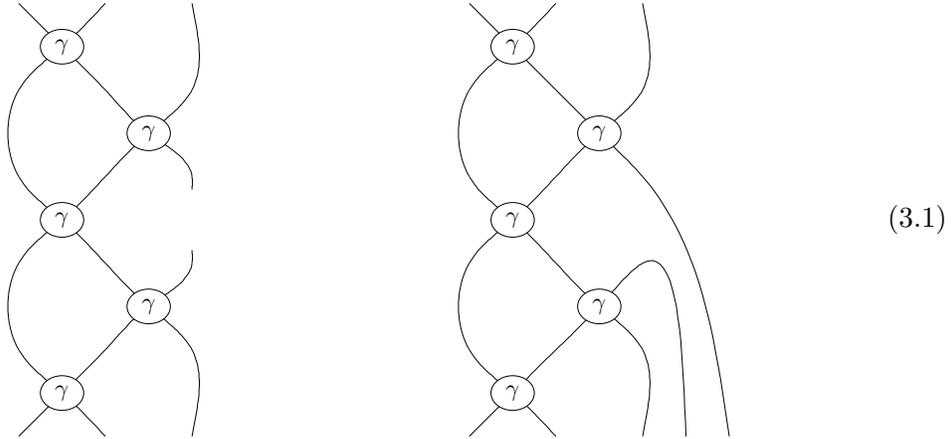
\begin{equation}
  \label{eq:ccp}
  \strid{gamma_cp_ctxt}
  \qquad\qquad\qquad\qquad
  \strid{gamma_cp_ctxt_compact}
\end{equation}
Of course, this ``diagram'' does not formally make sense: it is not a proper
string diagram in the usual sense~\cite{joyal-street:geometry-tensor-calculus},
because of the ``punched hole'' in the right border. However, one can make this
intuitive approach precise by embedding the 2-category of terms into the free
2-category with adjoints it generates, which string diagrammatically corresponds
to adding the possibility of ``bending wires''. The diagram on the left
of~\eqref{eq:ccp} will thus be actually formalized by a diagram such as the one
on the right.

This section certainly constitutes the most novel part of the present
paper. However, the purpose of this article was to introduce the reader to the
concept of higher-dimensional rewriting theory, and to motivate the further
developments which are described below. We detail the construction of the
multicategory of compact contexts generated by a 2\nbd{}polygraph~$E$ and show
that the 2-category~$E^*$ it generates can be embedded into it. We also describe
how this setting can be used in order to formulate a unification algorithm for
3\nbd{}polygraphs. A preliminary version of a formal exposition of these later
works can be found in~\cite{mimram:2-cp}.

\subsection{The multicategory of contexts}
\label{sec:contexts}
In order to formalize the notion of critical pair for a 3-polygraph, we need to
first formalize the notion of context in the category~$E^*$ generated by a
2-polygraph~$E$. This methodology can easily be generalized to $n$-polygraphs,
but we only describe it here in the case of dimension 3 for
clarity. Intuitively, a context is a morphism with multiple typed ``holes'' or
``metavariables''. Since such a context can have multiple inputs (\ie multiple
holes) and one output (the morphism resulting from filling the holes with
morphisms), the contexts are naturally structured as a
multicategory~\cite{leinster:higher-operads}.

\begin{definition}
  \label{def:multicat}
  A \emph{multicategory}~$\M$ (or \emph{colored operad}) is given by
  \begin{itemize}
  \item a class $\M_0$ of \emph{objects},
  \item a class $\M_1(A_1,\ldots,A_n;A)$ of \emph{$n$-ary operations} for every
    objects $A_1,\ldots,A_n$ and~$A$, we write $f:A_1,\ldots,A_n\to A$ to
    indicate that $f\in\M_1(A_1,\ldots,A_n;A)$,
  \item a \emph{composition} function which to every family of operations
    $f_i:A_i^1,\ldots,A_i^{k_i}\to A_i$, with \hbox{$1\leq i\leq n$}, and
    $f:A_1,\ldots,A_n\to A$, associates a composite operation
    \[
    f\circ(f_1,\ldots,f_n)
    \qcolon
    A_1^1,\ldots,A_1^{k_1},\ldots,A_n^1,\ldots,A_n^{k_n}\to A
    \]
    that we often simply write $f(f_1,\ldots,f_n)$,
  \item an operation $\id_A:A\to A$, called \emph{identity}, for every object~$A$,
  \end{itemize}
  such that
  \begin{itemize}
  \item composition is associative:
    \[
    \begin{array}{cl}
      &f\circ\pa{f_1\circ(f_1^1,\ldots,f_1^{k_1}),\ldots,f_n\circ(f_n^1,\ldots,f_n^{k_n})}\\[2ex]
      \qeq&
      \pa{f\circ(f_1,\ldots,f_n)}\circ(f_1,\ldots,f_1^{k_1},\ldots,f_n^1,\ldots,f_n^{k_n})

    \end{array}
    \]
    for every choice of operations $f$, $f_i$ and $f_i^j$ for which the
    compositions make sense,
  \item identities are neutral elements for composition: for every
    \hbox{$f:A_1,\ldots,A_n\to A$}, we have $f\circ(\id_A,\ldots,\id_A)=f$.
  \end{itemize}
  A \emph{symmetric multicategory} is a multicategory~$\M$ together with a
  bijection between the classes $\M(A_1,\ldots,A_n;A)$ and
  $\M(A_{\sigma(1)},\ldots,A_{\sigma(n)};A)$ of operations, for every
  permutation \hbox{$\sigma:n\to n$}, these bijections having to satisfy some
  coherence axioms that will be omitted here.
\end{definition}

Suppose given a 2-polygraph~$E$. Given a pair of 0-cells $A,B\in E_0^*$ and a
pair of parallel 1-cells $f,g:A\to B$ in~$E_1^*$, we write~$E[f\To g]$ for the
polygraph obtained from~$E$ by adding a new 2-generator~$X:f\To g$, \ie $(E[f\To
g])_2=E_2\uplus\{X\}$, such that~$s_1(X)=f$ and~$t_1(X)=g$. The 2-cells
in~$E[f\To g]^*$ being obtained by tensoring and composing 2\nbd{}generators
in~$E_2\uplus\set{X}$, they can be seen as terms in~$E^*$ with a
metavariable~$X$ of type~$f\To g$. In particular, any morphism of~$E^*$ can be
seen as a morphism of~$E[f\To g]^*$ which does not use the metavariable and we
write~$I_{f\To g}:E^*\to E[f\To g]^*$ for the corresponding inclusion functor.

By the universal property of the free 2-category~$E[f\To g]^*$ over the
polygraph~$E[f\To g]$, for any 2-category~$\D$, functor~$F:E^*\to\D$ and
2-cell~$\alpha:F(f)\To F(g)$ of~$\D$, there exists a unique
functor~$F[\alpha]:E[f\To g]^*\to\D$ such that
\[
F[\alpha]\circ I_{f\To g}=F
\qqtand
F(X)=\alpha
\]
where~$X$ is the newly added metavariable. A nice understanding of this can be
given by adopting a more abstract definition of
polygraphs~\cite{guiraud-malbos:higher-cat-fdt, mimram:2-cp}, that we did not
give here for the sake of simplicity. As a particular case, \emph{substitution}
can be defined using this property: given a 2-cell~$\alpha:f\To g$ in~$E_2^*$,
when~$\D=E^*$ and~$F:E^*\to E^*$ is the identity functor~$\Id$, the image of the
image~$\Id[\alpha](\beta)$ of a 2\nbd{}cell $\beta:h\To i$ of~$E_2[f\To g]^*$ is
denoted by~$\beta[\alpha]$, and corresponds to the morphism obtained
from~$\beta$ by replacing every instance of the metavariable~$X$ by~$\alpha$.

We more generally write~$E[f_1\To g_1,\ldots,f_n\To g_n]$ for $(((E[f_1\To
g_1])\ldots)[f_n\To g_n])$, extend notation for substitution accordingly, and
often write~$X_i$ for the newly introduced variable of type~$f_i\To g_i$. If we
fix an enumeration~$E_2=\{\alpha_1,\ldots,\alpha_n\}$ of the 2-generators
of~$E$, and write~$f_i$ (\resp $g_i$) for the source (\resp target)
of~$\alpha_i$, then the polygraph~$E$ is isomorphic to~$E'[f_1\To
g_1,\ldots,f_n\To g_n]$, where~$E'$ is the 2-polygraph obtained from~$E$ by
removing all the 2\nbd{}generators, \ie $E_2'=\emptyset$. Now, given a fixed
index~$k$, we write~$\phi:E_2\to\N$ for the function such
that~$\varphi(\alpha_i)=0$ if $i\neq k$, and~$\varphi(\alpha_k)=1$. If we
write~$\mathcal{N}$ for the 2-category with one 0-cell, one 1-cell, and~$\N$ as
set of 2-cells with addition as vertical composition and~$0$ as vertical
identity, and~$F:E'^*\to\mathcal{N}$ for the only functor between the
2-categories~$E'^*$ and~$\mathcal{N}$, the deduced
functor~$F[\varphi]:E^*\to\mathcal{N}$ sends a 2-cell~$\beta$ to an integer,
which is called the \emph{weight} of the 2-generator~$\alpha_k$ in~$\beta$ and
written~$\sizeof{\beta}_{\alpha_k}$. The weight $\sizeof{\beta}_\alpha$
indicates the number of times a 2-generator~$\alpha$ of~$E$ occurs in a
2-cell~$\beta$ of~$E^*$. Similarly, the \emph{size}~$\sizeof\beta$ of a
2-cell~$\beta$ is the number of 2-generators that it contains, \ie $\sizeof\beta=\sum_{\alpha\in E_2}\sizeof{\beta}_\alpha$.

\begin{definition}
  \label{def:context-multicat}
  The \emph{multicategory of contexts} of a 2-polygraph~$E$, denoted
  by~$\contexts{E}$, is the symmetric multicategory whose
  \begin{itemize}
  \item objects are pairs of parallel 1-cells $f$ and $g$ of~$E^*$, which are
    often written~$f\To g$,
  \item operations~$\alpha:f_1\To g_1\ ,\ \ldots\ ,\ f_n\To g_n\to f\To g$ are the
    2-cells~$\alpha$ of type~$f\To g$ in $E[f_1\To g_1,\ldots,f_n\To g_n]^*$
    such that for every index~$i$, $\sizeof{\alpha}_{X_i}=1$, \ie every
    metavariable occurs exactly once in the morphism,
  \end{itemize}
  composition is induced by substitution as expected, the identity~$\id_{f\To
    g}:f\To g\to f\To g$ on $f\To g$ is the variable~$X:f\To g$, and symmetry
  corresponds to renaming of variables.
\end{definition}

\begin{example}
  If we consider the 2-polygraph~$E$ of symmetries introduced in
  Example~\ref{ex:pres-bij}, the morphism
  \[
  (\gamma\otimes\id_1)\circ(\id_1\otimes\gamma)\circ(\id_1\otimes(X_2\circ X_1)\otimes\id_1)\circ(\id_1\otimes\gamma)\circ(\gamma\otimes\id_1)
  \]
  in~$E[1\To 1,1\To 1]$ is a morphism in~$\contexts{E}(1\To 1,1\To 1;3\To
  3)$. Graphically,
  \[
  \strid{ex_bij_holes}
  \]
\end{example}

In particular, if we restrict to unary operations (operations whose type is of
the form $f_1\To g_1\to f\To g$) the structure of multicategory reduces to a
structure of category of contexts, which acts on the set of 2-cells of the
2-category~$E^*$: if~$K:f_1\To g_1\to f\To g$ is a unary context,
and~$\alpha:f_1\To g_1$ is a 2-cell of~$\C$, we often write~$K(\alpha):f\To g$
for the 2\nbd{}cell~$K[\alpha]$ of~$\C$ obtained from~$K$ by substituting the
variable by~$\alpha$. More generally, a context in a 3-polygraph is a context in
the underlying 2-polygraph, and if $K:f_1\To g_1\to f\To g$ is a context in a
3-polygraph and $r:\alpha\TO\beta:f_1\To g_1$ is a 3-cell in the
3-category~$E^*$, we write~$K(r):K(\alpha)\To K(\beta):f\To g$ for the obvious
3-cell obtained from~$r$ by composing (in dimensions~$0$ and~$1$) with identity
2-cells. Notice also that the nullary operations are precisely the 2-cells
of~$E^*$.

It is possible to more generally define a notion of multicategory of
contexts~$\contexts\C$ of any 2\nbd{}category~$\C$, which coincides with the
previously given definition in the case where~$\C$ is of the form~$\C=E^*$ for
some 2-polygraph~$E$. We will only need to consider this last case in the
following, which is why we do not give the general definition.

\subsection{Critical pairs in 3-polygraphs}
\label{sec:cp-3}
Suppose fixed a 3-polygraph~$E$, freely generating a 3-category~$E^*$. Two
coinitial 3-cells
\[
r_1:\alpha\TO\beta_1:f\To g:A\to B
\qtand
r_2:\alpha\TO\beta_2:f\To g:A\to B
\]
of this 3-category are \emph{joinable} when there exists a 2-cell $\beta:f\To g$
and two 3\nbd cells $s_1:\beta_1\TO\beta$ and $s_2:\beta_2\TO\beta$ such that
$s_1\circ_2 r_1=s_2\circ_2 r_2$, where~$\circ_2$ denotes the composition in
dimension~$2$:
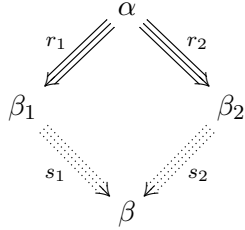
\[
\vxym{
  &\ar@3{->}[dl]_{r_1}\alpha\ar@3{->}[dr]^{r_2}&\\
  \beta_1\ar@3{.>}[dr]_{s_1}&&\beta_2\ar@3{.>}[dl]^{s_2}\\
  &\beta&
}
\]
Given a 3-generator~$r:\alpha\TO\beta$ and two 2-cells $\alpha'$ and~$\beta'$,
we write $\alpha'\TO^{K,r}\beta'$,
when there exists a unary context~$K$ such that~$K(\alpha)=\alpha'$ and
$K(\beta)=\beta'$. A polygraph is \emph{locally confluent} when for every cells
such that~$\alpha\TO^{K_1,r_1}\beta_1$ and~$\alpha\TO^{K_2,r_2}\beta_2$, the two
3-cells $K_1(r_1)$ and $K_2(r_2)$ are joinable. It is \emph{terminating} when
there is no infinite sequence
$\alpha_1\TO^{K_1,r_1}\alpha_2\TO^{K_2,r_2}\ldots$

The Newman lemma is valid in this
framework~\cite{guiraud-malbos:higher-cat-fdt}:

\begin{lemma}
  \label{lemma:newman-polygraphs}
  A terminating polygraph is confluent if and only if it is locally confluent.
\end{lemma}

In some simple cases, termination of polygraphs can be deduced from the fact
that all the rules make the size of morphisms decrease:
\begin{lemma}
  \label{lemma:size-term}
  If~$E$ is a 3-polygraph such that for every 3-generator $r:\alpha\TO\beta$ we
  have $\sizeof\alpha>\sizeof\beta$, then~$E$ is terminating.
\end{lemma}
\noindent
This simple criterion for showing the termination of a polygraph is often too
weak. More elaborate termination orders for 3-polygraphs have been studied by
Guiraud~\cite{guiraud:termination-3-rewr}. In this paper, we are mostly
interested in studying local confluence of polygraphs and will not detail those.
The usual notion of critical pair can be extended to the setting of 3-polygraphs
as follows.

\begin{definition}
  \label{def:unifier}
  A \emph{unifier} of a pair of 2-cells
  \[
  \alpha_1:f_1\To g_1
  \qtand
  \alpha_2:f_2\To g_2
  \]
  in a 2-category~$\C$ consists of a pair of cofinal unary contexts
  \[
  K_1:f_1\To g_1\to f\To g
  \qtand
  K_2:f_2\To g_2\to f\To g
  \]
  such that $K_1(\alpha_1)=K_2(\alpha_2)$. A unifier is \emph{most general}
  when it is
  \begin{itemize}
  \item \emph{non-trivial}: there is no binary context
    \[
    K
    \qcolon
    f_1\To g_1,f_2\To g_2\qto f\To h
    \]
    such that
    \[
    K_1=K\circ(\id_{f_1\To g_1},\alpha_2)
    \qtand
    K_2=K\circ(\alpha_1,\id_{f_2\To g_2})
    \]
  \item \emph{minimal}: for every unifier $(K_1', K_2')$ of~$\alpha_1$
    and~$\alpha_2$ such that $K_1=K_1''\circ K_1'$ and $K_2=K_2''\circ K_2'$ for
    some contexts~$K_1''$ and~$K_2''$, the unary contexts $K_1''$ and~$K_2''$
    are invertible.
  \end{itemize}
\end{definition}

\begin{definition}
  \label{def:cp}
  A \emph{critical pair} $(K_1,r_1,K_2,r_2)$ in a 3-polygraph~$\eqth{S}$
  consists of a pair of 3\nbd{}generators
  \[
  r_1:\alpha_1\TO\beta_1:f_1\To g_1
  \qtand
  r_2:\alpha_2\TO\beta_2:f_2\To g_2
  \]
  and a most general unifier
  \[
  K_1:f_1\To g_1\to f\To g
  \qtand
  K_2:f_2\To g_2\to f\To g
  \]
  of~$\alpha_1$ and~$\alpha_2$. We sometimes say that the 2-cell
  \hbox{$\alpha=K_1(\alpha_1)=K_2(\alpha_2)$} is a critical pair, by abuse of
  language.
\end{definition}

\begin{example}
  Consider the 3-polygraph~$E$ presenting the simplicial category~$\Delta$
  introduced in Example~\ref{ex:simpl-2-pres}. The morphisms
  \[
  \strid{rs_mon_unif2}
  \qquad\qquad\qquad\qquad
  \strid{rs_mon_unif1}
  \]
  are unifiers of the rules~$l$ and~$r$ which are not most general because they
  are respectively trivial and not minimal. The critical pairs (and thus the
  most general unifiers) have been described in Example~\ref{ex:simpl-2-pres}.
\end{example}

\subsection{Compact 2-categories}
As explained in the introduction of the present section, our aim is intuitively
to be able to consider diagrams such as the one on the left of~\eqref{eq:ccp} as
morphisms. Here, we achieve this by embedding the 2-category~$E^*$ generated by
a polygraph~$E$ into the compact 2\nbd{}category it freely generates. In these
categories, every 1-cell admits both a left and right adjoint, which graphically
essentially amounts to have the possibility to bend wires as in the right
of~\eqref{eq:ccp}.

The notion of adjunction in the $2$-category~$\Cat$ of categories, functors and
natural transformations can be generalized to any $2$-category~$\C$ as
follows~\cite{kelly-street:review-two-cat}. A $1$-cell $f:A\to B$ is \emph{left
  adjoint} to a $1$-cell~$g:B\to A$ (or $g$ is \emph{right adjoint} to~$f$),
that we will write $f\dashv g$, when there exist two $2$\nbd{}cells~$\eta:\id_A\To
f\otimes g$ and~$\varepsilon:g\otimes f\To\id_B$, called respectively the
\emph{unit} and the \emph{counit} of the adjunction and depicted respectively by
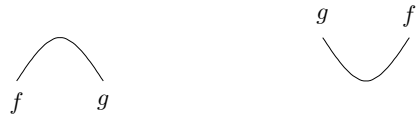
\[
\strid{adj_unit}
\qquad\qquad\qquad
\strid{adj_counit}
\]
such that $(f\otimes\varepsilon)\circ(\eta\otimes f)=\id_f$ and
$(\varepsilon\otimes g)\circ(g\otimes\eta)=\id_g$. These equations are called
the \emph{zig-zag laws} because of their graphical representation:
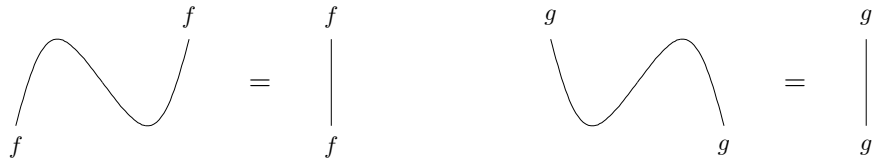
\begin{equation*}
  \strid{adj_zz_f_l}
  \qeq
  \strid{adj_zz_f_r}
  \qquad\qquad\qquad
  \strid{adj_zz_g_l}
  \qeq
  \strid{adj_zz_g_r}
\end{equation*}

The notion of 2-category with adjoints was studied in the case of symmetric
monoidal categories~\cite{kelly-laplaza:coherence-compact-closed} (where they
are called \emph{compact closed} categories), monoidal
categories~\cite{joyal-street:braided-tensor-categories} (where they are called
\emph{autonomous} categories), as well as other variants such as
\emph{spherical} categories~\cite{barrett-westbury:spherical-categories};
see~\cite{selinger:graphical-survey} for a concise presentation of those.

\begin{definition}
  A 2-category is \emph{compact} when every 1\nbd cell admits both a left and a
  right adjoint.
  A \emph{strictly compact} 2-category is a compact 2-category in which every
  1\nbd cell \hbox{$f:A\to B$} has an assigned left adjoint \hbox{$f^{-1}:B\to
    A$} and an assigned right adjoint \hbox{$f^{+1}:B\to A$}. We write
  $\eta_f^{+}$ and~$\varepsilon_f^{+}$ (\resp $\eta_f^{-}$ and
  $\varepsilon_f^{-}$) for the unit and the counit of the adjunction $f\dashv
  f^{+1}$ (\resp $f^{-1}\dashv f$). The following coherence axioms should
  moreover be satisfied:
  \begin{itemize}
  \item for every pair of composable 1-cells $f$ and $g$,
    \[
    (f\otimes g)^{-1}=g^{-1}\otimes f^{-1}
    \qquad\qquad
    (f\otimes g)^{+1}=g^{+1}\otimes f^{+1}
    \]
    and
    \[
    \eta_{f\otimes g}^+=(f\otimes\eta_g^+\otimes f^{+1})\circ\eta_f^+
    \qquad\qquad
    \varepsilon_{f\otimes g}^+=\varepsilon_g^+\circ(g^{+1}\otimes\varepsilon_f^+\otimes g)
    \]
     and
    \[
    \eta_{f\otimes g}^-=(g^{-1}\otimes\eta_f^-\otimes g)\circ\eta_g^-
    \qquad\qquad
    \varepsilon_{f\otimes g}^-=\varepsilon_f^-\circ(f\otimes\varepsilon_g^-\otimes f^{-1})
    \]
  \item for every 0-cell $A$,
    \[
    \id_A^{-1}=\id_A=\id_A^{+1}
    \]
    and
    \[
    \eta_{\id_A}^+=\id_A=\varepsilon_{\id_A}^+
    \qquad\qquad
    \eta_{\id_A}^-=\id_A=\varepsilon_{\id_A}^-
    \]
  \item for every 1-cell $f$,
    \[
    (f^{+1})^{-1}=f=(f^{-1})^{+1}
    \]
    and
    \[
    \eta_{f^{-1}}^+=\eta_f^-
    \qquad\qquad
    \varepsilon_{f^{-1}}^+=\varepsilon_f^-
    \qquad\qquad
    \eta_{f^{+1}}^-=\eta_f^+
    \qquad\qquad
    \varepsilon_{f^{+1}}^-=\varepsilon_f^+
    \]
  \end{itemize}

  For any 1-cell~$f:A\to B$ in a strictly compact 2\nbd{}category and
  integer~$n$, the morphism~$f^n$ denotes the morphism defined by~$f^0=f$,
  $f^{n+1}=(f^n)^{+1}$ and~$f^{n-1}=(f^n)^{-1}$. We also simply write
  $\eta_f:B\To f^{-1}\otimes f$ and \hbox{$\varepsilon_f:f\otimes f^{-1}\To A$}
  for the unit and the counit of the adjunction between $f^{-1}$ and $f$.
\end{definition}

In the following, we suppose for simplicity that all the compact categories we
consider are equipped with a structure of strictly compact category. This is not
restrictive since every compact 2-category can be shown to be equivalent to a
strict one using an argument similar to the coherence theorem for compact closed
categories~\cite{kelly-laplaza:coherence-compact-closed}. The category of
compact categories is denoted by~$\CCat$.

\subsection{Embedding 2-categories into compact 2-categories}
\label{sec:compact-embedding}
There is an obvious forgetful functor from the category of compact 2-categories
to the category of 2-categories, and this forgetful functor admits a left
adjoint. We write $\compact{\C}$ for the free compact 2-category on a
2-category~$\C$ (the $\mathcal{A}$ here stands for ``adjoints''). The
construction of this free 2-category is detailed
in~\cite{preller-lambek:free-compact} and consists essentially in adapting the
work of Kelly and Laplaza on compact closed
categories~\cite{kelly-laplaza:coherence-compact-closed} to monoidal categories
which are not supposed to be symmetric. We recall briefly this construction
here.

The underlying category of a compact 2-category is naturally equipped with a
structure of ``category with formal adjoints'' in the following sense:

\begin{definition}
  A \emph{category with formal adjoints} $(\C,(-)^{-1},(-)^{+1})$ is a category
  together with two functors
  \[
  (-)^{-1}:\C\to\C^\op
  \qtand
  (-)^{+1}:\C^\op\to\C
  \]
  such that $((-)^{-1})^{+1}=\id_\C$ and $((-)^{+1})^{-1}=\id_{\C^\op}$.
\end{definition}

Given a 2-category~$\C$, with~$\cattrunc\C1$ as underlying category, the
underlying category of~$\compact{\C}$ is the free category with formal adjoints
on~$\cattrunc\C1$. More concretely, this category is the free category on the
graph whose objects are the objects of~$\cattrunc\C1$ as objects and whose
arrows $f^n:A\to B$ are pairs constituted of an integer $n\in\Z$, called a
\emph{winding number}, and a morphism $f:A\to B$ in~$\C$ if~$n$ is even (\resp a
morphism $f:B\to A$ in~$\C$ if~$n$ is odd), quotiented by the following
equalities:
\begin{itemize}
\item for every pair of composable morphisms $f^n$ and $g^n$,
  \[
  f^n\otimes g^n \qeq
  \begin{cases}
    (f\otimes g)^n&\text{if $n$ is even}\\
    (g\otimes f)^n&\text{if $n$ is odd}
  \end{cases}
  \]
\item for every object~$A$,
  \[
  (\id_A)^n
  \qeq
  \id_A
  \]
\end{itemize}
The 2-cells of~$\compact{\C}$ are formal vertical and horizontal composites of
\begin{itemize}
\item $\alpha^0:f^0\To g^0$, where $\alpha:f\To g$ is a 2-cell of~$\C$,
\item $\eta_{f^n}:\id_B\To f^{n-1}\otimes f^n$, for every 1-cell $f^n:A\to B$,
\item $\varepsilon_{f^n}:f^n\otimes f^{n-1}\To\id_A$, for every 1-cell $f^n:A\to B$,
\end{itemize}
quotiented by
\begin{itemize}
\item the axioms of 2-categories (see Definition~\ref{def:2-cat}),
\item for every pair of vertically composable 2-cells $\alpha^0$ and $\beta^0$,
  \[
  \beta^0\circ\alpha^0
  \qeq
  (\beta\circ\alpha)^0
  \]
\item for every 1-cell $f^0$,
  \[
  \id_{f^0}
  \qeq
  (\id_f)^0
  \]
\item for every pair of horizontally composable 2-cells $\alpha^0$ and
  $\beta^0$,
  \[
  \alpha^0\otimes\beta^0
  \qeq
  (\alpha\otimes\beta)^0
  \]
\item for every 1-cell $f^n$,
  \[
  (f^{n-1}\otimes\varepsilon_{f^n})\circ(\eta_{f^n}\otimes f^{n-1})=f^{n-1}
  \qtand
  (\varepsilon_{f^n}\otimes f^n)\circ(f^n\otimes\eta_{f^n})=f^n
  \]
\end{itemize}
Graphically, if we write respectively
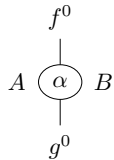
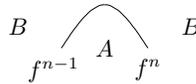
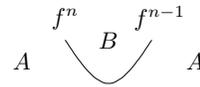
\[
\strid{compact_f}
\qquad\qquad
\strid{compact_eta}
\qquad\qquad
\strid{compact_eps}
\]
for $\alpha^0:f^0\To g^0:A\to B$, $\eta_{f^n}$ and $\varepsilon_{f^n}$ (where
$f^n:A\to B$), the four last equalities can be pictured as follows
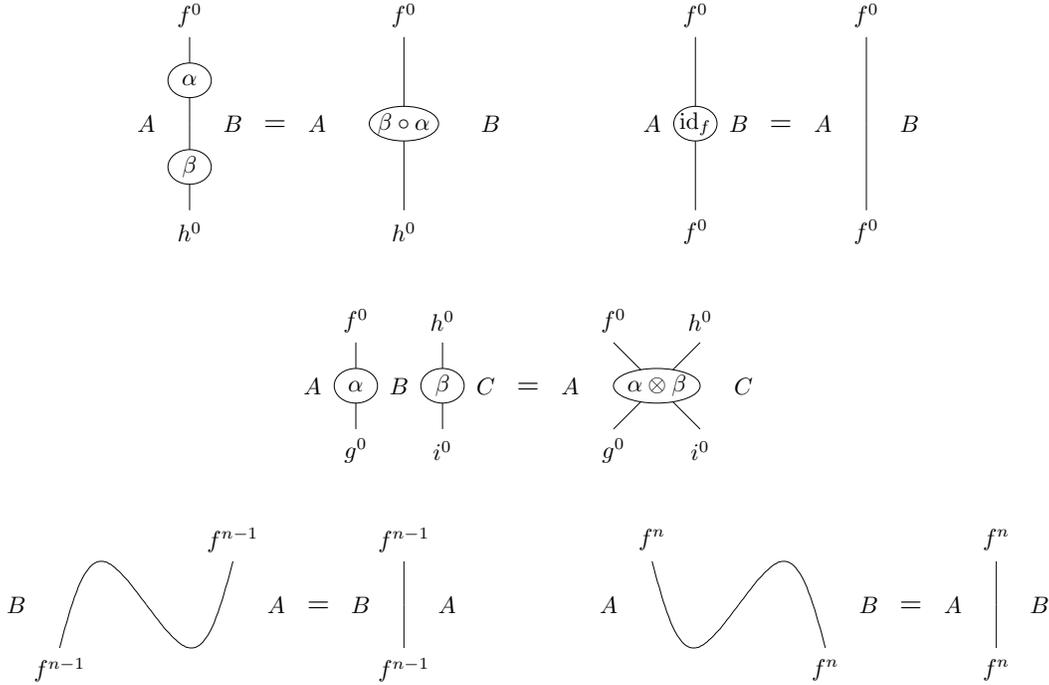
\[
\begin{array}{c}
  \strid{compact_v_l}
  =
  \strid{compact_v_r}
  \qquad\qquad
  \strid{compact_id_l}
  =
  \strid{compact_id_r}
  \\
  \strid{compact_h_l}
  =
  \strid{compact_h_r}
  \\
  \strid{compact_adjl_l}
  =
  \strid{compact_adjl_r}
  \qquad\qquad
  \strid{compact_adjr_l}
  =
  \strid{compact_adjr_r}
\end{array}
\]
The string diagrams for compact categories are studied
in~\cite{joyal-street:planar-diagrams-tensor-algebra}.

\begin{remark}
  If~$\C$ is the 2-category~$E^*$ generated by a 2-polygraph~$E$, the compact
  2\nbd{}category~$\compact{\C}$ is presented by the 3-polygraph~$F$ such that
  \begin{itemize}
  \item $F_0=E_0$
  \item $F_1=\setof{f^n}{f\in E_1, n\in\Z}$
  \item $F_2=\setof{\alpha^0}{\alpha\in E_2}\uplus\setof{\eta_{f^n},\varepsilon_{f^n}}{f^n\in F_1}$
  \item $F_3=\setof{l_{f^n},r_{f^n}}{f^n\in F_1}$
  \end{itemize}
  with
  \[
  \begin{array}{r@{\qcolon}r@{\qTO}l}
    l_{f^n}&(f^{n-1}\otimes\varepsilon_{f^n})\circ(\eta_{f^n}\otimes f^{n-1})&f^{n-1}\\
    r_{f^n}&(\varepsilon_{f^n}\otimes f^n)\circ(f^n\otimes\eta_{f^n})&f^n
  \end{array}
  \]
  and other cells have the obvious source and target. By
  Lemma~\ref{lemma:size-term}, the polygraph~$F$ is terminating and by
  Lemma~\ref{lemma:newman-polygraphs} it is confluent since all its critical
  pairs, which are of the form
  \[
  \strid{compact_cp1}
  \qtand
  \strid{compact_cp2}
  \]
  for some 1-cell $f^n$, are joinable.
\end{remark}

\begin{lemma}
  With the notations of the preceding remark, if $f_1,\ldots,f_m$ and
  $g_1,\ldots,g_n$ are parallel lists of composable morphisms of~$E^*$,
  then the 2-cells
  \begin{equation*}
    \label{eq:compact-embedding-cell}
    \alpha
    \qcolon
    f_1^0\otimes\ldots\otimes f_m^0
    \qTo
    g_1^0\otimes\ldots\otimes g_n^0
  \end{equation*}
  in the underlying 2-category of~$F$ which are normal forms (\wrt the rewriting
  rules of~$F$) do not contain any 2-generator~$\eta_{f^k}$
  or~$\varepsilon_{f^k}$.
\end{lemma}
\begin{proof}
  It is easy to show that a 2-cell~$\alpha$ in~$F^*$ can be written as a
  composite of morphisms of the form~$\id_g\otimes\beta\otimes\id_h$
  where~$\beta$ is either a 2-cell of~$\C$ or a morphism of the
  form~$\eta_{f^k}$ or~$\varepsilon_{f^k}$ (see for
  example~\cite{lafont:boolean-circuits}). Suppose that~$\alpha$ contains a
  2-generator of the form~$\varepsilon_{f^k}$ with~$k>0$. It can therefore be
  written as a composite of the form
  \[
  \strid{embedding_alpha_1}
  \]
  The 2-cells~$\alpha_1$ and~$\alpha_2$ are noted with boxes for clarity and $0$
  stands for a 1-cell whose winding number is~$0$. Since the only generators
  whose target contain a 1-cell of the form~$f^k$ are~$\eta_k$ and~$\eta_{k+1}$,
  the 2-cell~$\alpha_1$ is necessarily of one of the three following forms:
  \[
  \hspace{-4ex}
  \strid{embedding_alpha_1_1}
  \]
  or
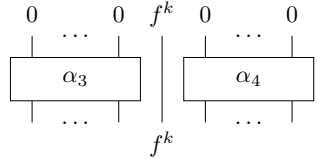
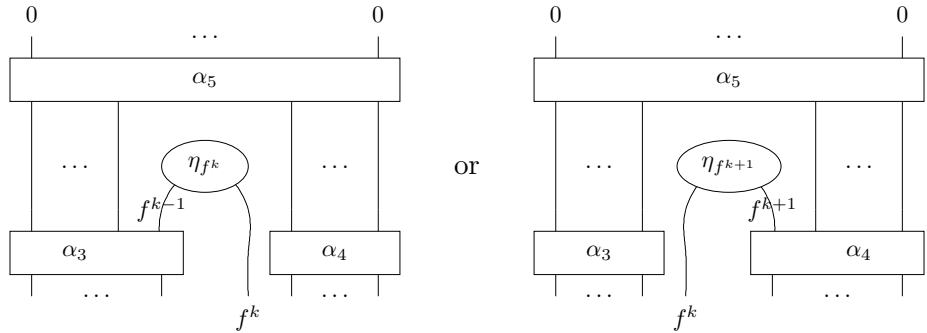
  \[
  \strid{embedding_alpha_1_2}
  \tor
  \strid{embedding_alpha_1_3}
  \]
  The first case is impossible since $k>0$ and we have supposed that all the
  winding numbers of the 1-generators occurring in the source of~$\alpha$
  are~$0$. The second case is not possible either since the morphism $\alpha$
  would not be a normal form (the rule~$l_{f^{k-1}}$ could be applied), and the
  2-cell~$\alpha$ thus contains a 2-generator~$\eta_{f^{k+1}}$. By using a
  similar argument, $\alpha$ also contains the
  2-generator~$\varepsilon_{f^{k+2}}$. So, by induction, the 2-cell~$\alpha$
  would contain all the 2\nbd{}generators~$\varepsilon_{f^{k+2i}}$ with~$i\in\N$
  and would therefore be a composite of an infinite number of generators. This
  is absurd since the 2-cells in~$F^*$ are inductively generated. We deduce
  that~$\alpha$ does not contain a 2-generator of the form~$\varepsilon_{f^k}$
  with~$k>0$. Similarly, it does not contain a 2-generator~$\eta_{f^k}$
  with~$k>0$. The cases where $k\leq 0$ are also similar (we construct an
  infinite sequence of generators that~$\alpha$ would contain, with strictly
  decreasing winding numbers).
\end{proof}
From this, we can deduce that the 2-cells
\[
\alpha
\qcolon
f_1^0\otimes\ldots\otimes f_m^0
\qTo
g_1^0\otimes\ldots\otimes g_n^0
\]
in~$F^*$ are in bijection with the 2-cells
\[
\alpha
\qcolon
f_1\otimes\ldots\otimes f_m
\qTo
g_1\otimes\ldots\otimes g_n
\]
of~$\C$, which shows that the embedding of $E^*$ into $F^*$ is full and faithful
(this embedding is the functor defined as the identity on objects, as~$f\mapsto
f^0$ on 1-cells and as~$\alpha\mapsto\alpha^0$ on 2\nbd{}cells). A
2-cell~$\alpha$ whose source and target is of the form above is called
\emph{regular} (the regular 2-cells are thus those which are in the image of the
embedding). The argument can easily be generalized to any category~$\C$, not
necessarily generated by a 2\nbd{}polygraph (but we will only make use of the
case proved in previous lemma):

\begin{proposition}
  \label{prop:compact-embedding}
  The components $\eta_\C:\C\to\compact{\C}$ of the unit of the adjunction
  between 2-categories and compact 2-categories are full and faithful.
\end{proposition}

\noindent
This means that given two 0-cells~$A$ and~$B$ of~$\C$, the
hom-categories~$\C(A,B)$ and~$\compact\C(A,B)$ are isomorphic in a coherent
way. The 2-category~$\compact\C$ thus provides a ``larger world'' in which we
can embed the 2-category~$\C$ without losing information.

An interesting observation about compact 2-categories is that, in those, the
distinction between the source and the target of a 1-cell is quite
``artificial''. This is formalized by the following proposition.

\begin{proposition}
  \label{prop:compact-iso}
  If~$\C$ is a compact 2-category, the sets
  \[
  \Hom(f\otimes g,h)
  \qcong
  \Hom(g,f^{-1}\otimes h)
  \]
  are naturally isomorphic by the function
  \[
  \alpha
  \qmapsto
  (f^{-1}\otimes\alpha)\circ(\eta_{f}\otimes g)
  \]
  Graphically,
  \[
  \strid{compact_iso_l}
  \quad\mapsto\quad
  \strid{compact_iso_r}
  \]
  Similarly, the sets
  \[
  \Hom(f\otimes g,h)
  \qcong
  \Hom(f,h\otimes g^1)
  \]
  are naturally isomorphic by the function
  \[
  \alpha
  \qmapsto
  (\alpha\otimes g^1)\circ(f\otimes\eta_{g^1})
  \]
  These bijections are called \emph{rotations}.
\end{proposition}
In particular, for any pair of 1-cells $f,g:A\to B$, the set $\Hom(f,g)$ is
isomorphic to $\Hom(\id_B,f^{-1}\otimes g)$ and to~$\Hom(f\otimes
g^{-1},\id_A)$. This shows that the notion of ``input'' and ``output'' of
2-cells is fairly artificial in compact 2-categories. Actually, based on these
ideas, it is possible to reformulate equivalently the notion of compact
2-category in a way such that 2-cells only have one ``border'' instead of having
both a source and a target. We have called the resulting notion a \emph{rotative
  2-category}, details can be found in~\cite{mimram:2-cp}.

\bigskip
Proposition~\ref{prop:compact-embedding} provides us with a full and faithful
embedding~$\C\to\compact\C$ of a 2-category~$\C$ into the compact 2-category
that it freely generates.  The interest of this embedding is that there are
``extra morphisms'' in~$\compact\C$ that can be used to represent ``partial
compositions'' in~$\C$. For example, consider two 2-cells~$\alpha:f\To
f_1\otimes g\otimes f_2$ and $\beta:h_1\otimes g\otimes h_2\To h$ in~$\C$. These
can be seen as the morphisms of~$\compact\C$ depicted on the left
of~\eqref{eq:partial-comp} by the previous embedding.

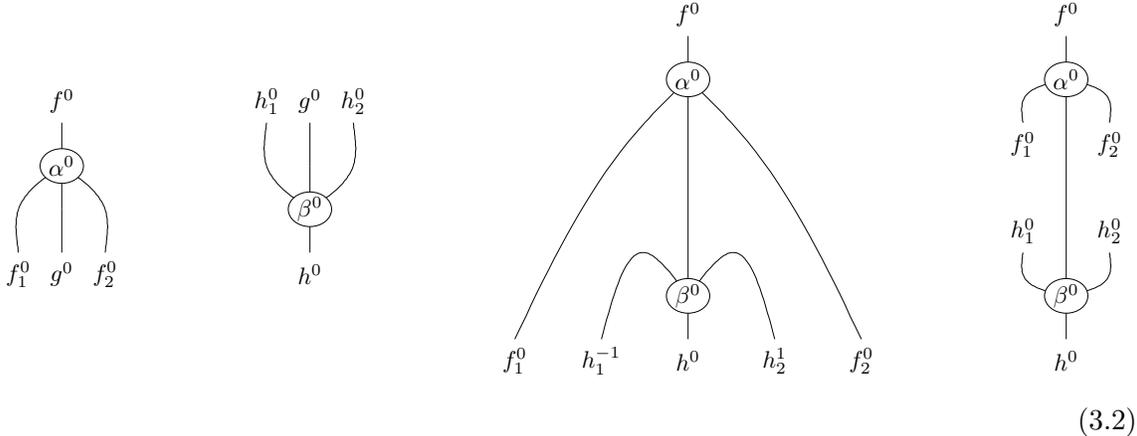
\begin{equation}
  \label{eq:partial-comp}
  \strid{pcomp_a}
  \qquad\qquad
  \strid{pcomp_b}
  \qquad\qquad
  \strid{pcomp_rot}
  \qquad\qquad
  \strid{pcomp}
\end{equation}
From these two morphisms, the morphism \hbox{$\alpha\otimes_g\beta:f^0\To
  f_1^0\otimes h_1^{-1}\otimes h^0\otimes h_2^1\otimes f_2^0$}, depicted in the
center right of~\eqref{eq:partial-comp}, can be constructed. This morphism
represents the \emph{partial composition} of the 2-cells~$\alpha$ and~$\beta$ on
the 1-cell~$g$: up to rotations, this 2-cell is fundamentally a way to give a
precise meaning to the diagram depicted on the right of~\eqref{eq:partial-comp}.

\subsection{Compact polygraphs}
\label{sec:compact-pol}
The notion of \emph{compact 2-polygraph}~$E$ is defined as in~\eqref{eq:rs-3},
where~$\fpoly{E_1}$
is now the set of
morphisms of the free category with formal adjoints on the
graph~\eqref{eq:rs-1}, and $E_2^*$ is the set of 2-cells of the free compact
2-category with the previously generated category as underlying category with
formal adjoints. We write~$\ncPol{2}$ for the category of compact
$2$-polygraphs.

Similarly, the construction of the multicategory of contexts given in
Section~\ref{sec:contexts} can be straightforwardly adapted to compact
polygraphs. Given a compact 2-polygraph~$E$, we still write~$\contexts E$ for
the multicategory thus generated, whose operations are then called \emph{compact
  contexts}. The setting of compact contexts provides a generalization of
partial composition by allowing a ``partial composition of a morphism with
itself''. Namely, from a context
\hbox{$\alpha:(\ldots,(f_i,g_i),\ldots)\to(f,g^1\otimes h\otimes g^0)$} with
\hbox{$f:A\to A$} and~$h:B\to B$ one can build the context
\[
\varepsilon_g^0\circ(g^1\otimes X\otimes g^0)\circ\alpha
\qcolon
(\ldots,(f_i,g_i),\ldots,(h,\id_B))
\qto
(f,\id_A)
\]
where \hbox{$X:h\to\id_B$} is a fresh variable. Graphically,
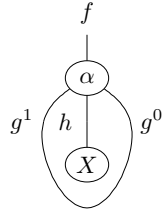
\[
\strid{merge}
\]
This operation amounts to \emph{merging} the outputs of type~$g^1$ and~$g^0$
of~$\alpha$, thus creating a ``hole'' which is formally taken in account as a
metavariable in the multicategory of compact contexts.

\subsection{Towards a unification algorithm for 3-polygraphs}
\label{sec:algo}
Using the partial composition and merging operations defined above, it is
possible to formulate an algorithm for computing critical pairs in a
3-polygraph, which will proceed as follows. Suppose given a
3-polygraph~$E$. This polygraph generates a 3-category~$E^*$ whose underlying
2-category is written~$\C$ (the 2-category $\C$ is freely generated by the
underlying 2\nbd{}polygraph~$\polytrunc{E}{2}$ of~$E$). The 2-category~$\C$ can
be fully and faithfully embedded into the free compact 2-category~$\compact\C$
it generates (Proposition~\ref{prop:compact-embedding}).
In turn, this compact 2-category generates a multicategory of
contexts~$\contexts{\compact\C}$ in which it can be embedded (any 2-cell
of~$\compact\C$ can be seen as a nullary context
in~$\contexts{\compact\C}$). Given two rewriting rules
\[
r':\alpha'\TO\beta':f'\To g'
\qqtand
r'':\alpha''\TO\beta'':f''\To g''
\]
the two 2-cells~$\alpha':f'\To g'$ and~$\alpha'':f''\To g''$ in~$\C$ can be seen
as 2-cells~$\alpha'^0:f'^0\To g'^0$ and~$\alpha''^0:f''^0\To g''^0$
in~$\compact\C$, which in turn can be seen as nullary contexts
in~$\contexts{\compact\C}(;f'^0\To g'^0)$ and~$\contexts{\compact\C}(;f''^0\To
g''^0)$ respectively. Our algorithm will compute a unifier of those in the
category of compact contexts, consisting of a pair
\[
K'
\qcolon
f'^0\To g'^0,f_1\To g_1,\ldots,f_n\To g_n
\qto
f\To g
\]
and
\[
K''
\qcolon
f''^0\To g''^0,f_1\To g_1,\ldots,f_n\To g_n
\qto
f\To g
\]
of compact contexts such that
\begin{equation}
  \label{eq:compact-unifier}
  K'(\alpha'^0,\id_{f_1\To g_1},\ldots,\id_{f_n\To g_n})
  \qeq
  K''(\alpha''^0,\id_{f_1\To g_1},\ldots,\id_{f_n\To g_n})
\end{equation}
which is minimal and non-trivial (in a sense similar to
Definition~\ref{def:unifier}), up to the symmetry of the multicategory and
rotations. Notice that the usual critical pairs in 3-polygraphs having
2-generators of arity~$1$ (\eg the polygraph of monoids given in
Example~\ref{ex:simpl-2-pres}) are recovered as the particular case where the
unifiers are such that~$n=0$ (there is no hole in the unifier) and the
morphism~\eqref{eq:compact-unifier} is regular, \ie the
morphism~\eqref{eq:compact-unifier} is in the image of the embedding of~$\C$
into~$\contexts{\compact\C}$. When it is not the case, all the unifiers
of~$\alpha'$ and~$\alpha''$ (in the sense of Definition~\ref{def:unifier}) can
be recovered as the contexts of the form
\[
K\circ K'(\id_{f'\To g'},\alpha_1,\ldots,\alpha_n)
\qqtand
K\circ K'(\id_{f''\To g''},\alpha_1,\ldots,\alpha_n)
\]
where the~$\alpha_i$ are morphisms in~$\compact\C$ seen as nullary compact
contexts and $K$ is a unary regular compact context. In this sense, the critical
pairs in the category of compact contexts generate all the unifiers in the usual
sense, and have the advantage of always being finite in number for a finite
polygraph.


A concrete description of the algorithm is out of the scope of this paper since
it requires the elaboration of an explicit representation of the morphisms in the
2-category~$E^*$ generated by a 2-polygraph~$E$~\cite{mimram:critical-pairs,
  mimram:2-cp}: up to now we have defined these morphisms using an abstract
universal construction (see Section~\ref{sec:rs-free-2-cat}) however a more
concrete representation is needed in order manipulate them algorithmically. Such
a representation was developed by the author, by describing the 2-cells in~$E^*$
themselves as polygraphs labeled by~$E$ (\ie objects in the slice category
$\slicecat{\nPol{2}}{E}$) up to isomorphism. For example, if~$E$ is the
2-polygraph corresponding to the signature of monoids defined in
Example~\ref{ex:simpl-2-pres}, the morphism $\mu\circ(\mu\otimes\id_1)$ can be
represented by the polygraph~$M$ such that
\[
M_0=\{1,\ 2,\ 3,\ 4\}
\qquad
M_1=\{5:1\to 2,\ 6:2\to 3,\ 7:3\to 4,\ 8:1\to 3,\ 9:1\to 4\}
\]
and
\[
M_2=\{10:5\otimes 6\To 8,\ 11:8\otimes 7\To 9\}
\]
(we have chosen to number the generators in an arbitrary order) labeled by the
morphism of polygraphs~$\lambda:M\to E$ such that all the 0-, 1- and
2-generators are labeled by~$*$, $1$ and~$\mu$ respectively. Graphically, this
corresponds to give different names (or numbers in this case) to the various
instances of generators used to build the morphism:
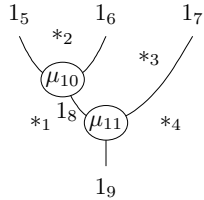
\[
\strid{mumu_repr}
\]
Of course, the naming of generators is not canonical which explains why we have
to consider these labeled polygraphs up to isomorphism (this can be seen as
some form of $\alpha$-conversion). The precise description of this 2-category
can be achieved by constructing a structure of monoidal globular
category~\cite{batanin:mon-glob-cat} on the globular category of
polygraphs. This will be presented in detail in further work.

\bigskip

We shall only illustrate how our algorithm works, by giving an example. Consider
a 3\nbd{}polygraph~$E$ with one 0-generator $*$, one 1-generator $1:*\to *$,
three 2-cells $\delta:1\to 4$, $\mu:4\to 1$ and $\sigma:1\to 1$ (where $4$
denotes $1\otimes 1\otimes 1\otimes 1$) and two rewriting rules (3-generators)
whose left members are respectively~$\alpha=\varsigma\circ\delta$ and
$\beta=\mu\circ\varsigma$, where
$\varsigma=\sigma\otimes\sigma\otimes\sigma\otimes\sigma$, \ie string
diagrammatically
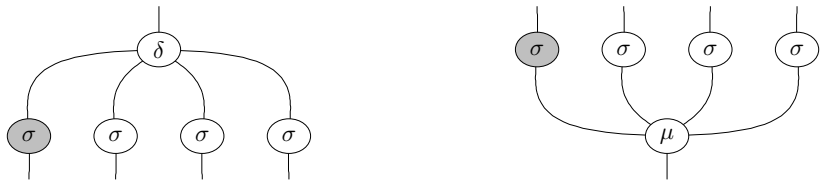
\[
\strid{unif_a}
\qquad\qquad
\strid{unif_b}
\]
The algorithm for unifying these two morphisms will go on as follows.
It is non-deterministic and all the unifiers of the two morphisms
(possibly with duplicates) will be obtained as the collection of all the results
of non-failed execution branches the algorithm. It starts by choosing
a 2-generator in each of the morphisms (the gray-colored ones in the diagram
above), and then it will progressively attach new cells to the left one, or link
wires together so that it becomes a most general unifier. If the two chosen
2\nbd{}generators are not the same then the algorithm fails (here it does not
because both are~$\sigma$). Since the~$\sigma$ generator selected on the
morphism on the left is linked with a~$\mu$ generator, the algorithm will start
by adding a~$\mu$ generator to the morphism on the left and doing a partial
composition with it, as shown on the left of~\eqref{ex:unif}. The unification is
then propagated: the~$\mu$ generator in the morphism on the right is connected
with four~$\sigma$ generators. The algorithm will non-deterministically choose
one which has not already been unified and will propagate the unification. For
example, if the third~$\sigma$ from the left is selected, the morphism in the
middle of~\eqref{ex:unif} might be obtained by adding a new~$\sigma$ generator
to the morphism. Non-deterministically, instead of adding a new generator, the
algorithm might choose to merge (using the operation described at the end of
Section~\ref{sec:compact-pol}) two inputs or outputs of the morphisms and the
morphism on the right of~\eqref{ex:unif} might be obtained as well (in this case
a new hole is added to the unifier being constructed).
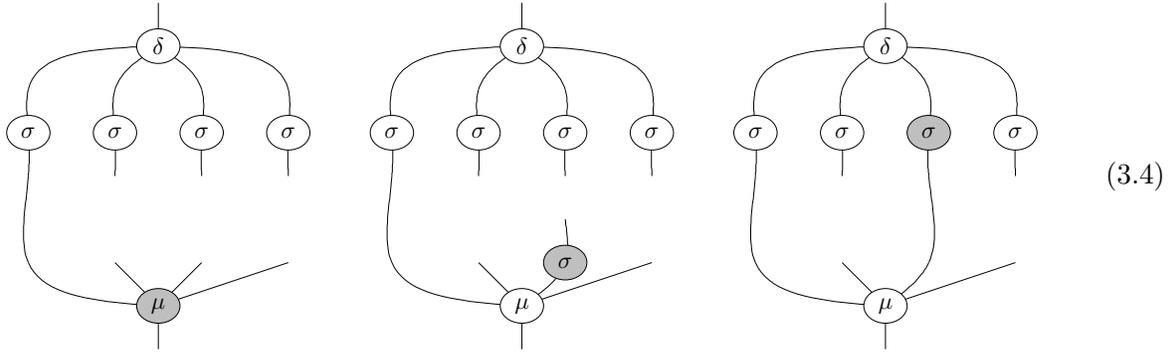
\begin{equation}
  \label{ex:unif}
  \hspace{-3ex}
  \strid{unif_ex1}
  \qquad
  \strid{unif_ex2}
  \qquad
  \strid{unif_ex3}
\end{equation}
Finally, by fully executing the algorithm, the three morphisms below will be
obtained as unifiers (as well as many others).
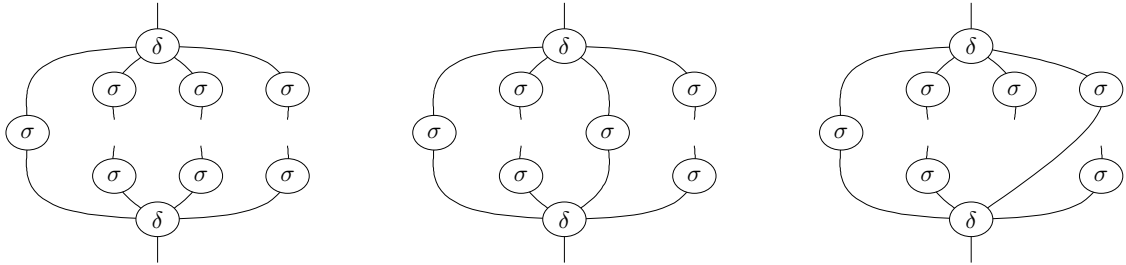
\begin{equation*}
  \hspace{-3ex}
  \strid{unif_ex_res1}
  \enspace
  \strid{unif_ex_res2}
  \enspace
  \strid{unif_ex_res3}
\end{equation*}

It can be shown that the algorithm terminates and generates all the critical
pairs in compact contexts, and these are in finite number. It is important to
notice that the algorithm generates the critical pairs of a rewriting system~$R$
in the ``bigger world'' of compact contexts, from which we can generate the
critical pairs in the 2-category generated by~$R$ (which are not necessarily in
finite number as explained in the introduction). If joinability of the critical
pairs in compact contexts implies that the rewriting system is confluent, the
converse is unfortunately not true: a similar situation is well known in the
study of $\lambda$-calculus with explicit substitutions, where a rewriting system
might be confluent without being confluent on terms with metavariables.

We have done a toy implementation of the algorithm in less than 2000 lines of
OCaml, with which we have been able to successfully recover the critical pairs
of rewriting systems in~\cite{lafont:boolean-circuits}. Even though we did not
particularly focus on efficiency, the execution times are good, typically less
than a second, because the morphisms involved in polygraphic rewriting systems
are usually small (but they can generate a large number of critical pairs). We
thus have hope to be able to build efficient tools in order to help dealing with
large algebraic structures.

\section{Future work}
\label{sec:future-work}
We have tried gradually to expose the notion of higher-dimensional rewriting
system and to connect it with the well known and well studied special cases of
string and term rewriting system. We have also introduced the notion of
multicategory of compact contexts generated by a 2-polygraph, which lays the
theoretical foundations for unification in polygraphic 2\nbd{}dimensional
rewriting systems. This leaves many research tracks open for future work, some
of which are detailed below.

\subsection{A 2-dimensional unification algorithm}
We have hinted in Section~\ref{sec:algo} how the theoretical tools introduced in
this article can be used to formulate a unification algorithm for
3-polygraphs. This algorithm will be described and proven correct in detail in
subsequent work~\cite{mimram:2-cp}. In particular, we plan to study the precise
links between our algorithm and the usual unification for term rewriting
systems, as well as algorithms for (planar) graph rewriting.

\subsection{Compact rewriting systems}
The use of compact 2-categories seems to be very promising, since it provides a
bigger world in which unification is simple to handle (there are a finite number
of critical pairs in particular). Moreover, left and right members of rules in
polygraphic rewriting systems are morphisms in 2-categories, but we can extend
the framework to have ``compact rewriting rules'' whose left and right members
are morphisms in compact 2-categories. There is no known finite convergent
polygraphic rewriting system presenting the category~$\Rel$ of finite sets and
relations~\cite{lafont:boolean-circuits} (which corresponds to the theory of
qualitative bicommutative bialgebras~\cite{mimram:first-order-causality}). We
conjecture that such a system does not exist. However, we believe that it would
be possible to have a finite convergent compact polygraphic rewriting system
containing rules such as
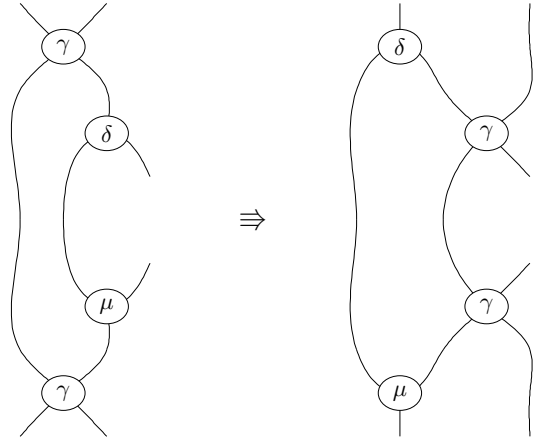
\[
\strid{rel_ex_l}
\qqTO
\strid{rel_ex_r}
\]
where~$\gamma$ is the generator for the symmetry,~$\delta$ is the
comultiplication and~$\mu$ is the multiplication. We plan to use our unification
algorithm in order to define and study such a rewriting system. It would also be
interesting to adapt the techniques developed by Guiraud to show termination of
polygraphic systems~\cite{guiraud:termination-3-rewr}.

\subsection{Parametric polygraphs}
In order to describe those free compact 2\nbd categories, we had to modify the
definition of the notion of polygraph by replacing the free category
construction by a free category with formal adjoints construction, and the free
2-category construction by a free compact 2-category construction. This suggests
that it might be interesting to investigate a more modular notion of polygraph,
parametrized by a series of adjunctions, which could be used to generate free
$n$-categories \emph{with properties} (e.g.~compact categories, groupoids,
etc.).

\subsection{Towards higher dimensions}
Since the notion of polygraphic rewriting system can be generalized to any
dimension, we would like to also have a generalization of rewriting theory to
higher dimensions using polygraphic rewriting systems. This would require a more
abstract and general formulation of the unification techniques that are used
here, in order to be able to extend them easily to higher dimensions.


\subsection{Practical uses of this work}
In some sense, our work can be considered as an algebraic study of the notion of
a bunch of operators linked by planar wires. We believe that this point of view
should be taken seriously and we plan to investigate a possible application of
the polygraphic rewriting techniques to electronic circuits. This could provide
a nice theoretical framework in which we could express and study optimization of
integrated circuits. Another field of application should be the design of an
optimizing language for digital signal processing. Sound effects are often
described by diagrammatic notation which is very close to the
string-diagrammatic notation for morphisms in 2-categories generated by
2-polygraphs. For instance an echo operator is often pictured as
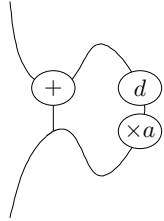
\[
\strid{echo}
\]
where~$+$ adds to signals, $\times a$ amplifies the signal by a
coefficient~$a<1$ and~$d$ delays the signal for~$d$ seconds: the echo is
obtained by adding the signal~$d$ seconds before, at a lower volume, to the
current signal. The rewriting techniques offered by polygraphs could therefore
be used in order to optimize those circuits using 3-dimensional rewriting
systems.


\subsection*{Acknowledgments}
The author is much indebted to John Baez, Albert Burroni, Jonas Frey, Emmanuel
Haucourt, Martin Hyland, Yves Lafont, Paul-André Melliès and François Métayer.

\bibliography{papers}
\bibliographystyle{abbrv}
\end{document}